\documentclass[sigconf]{acmart}

\usepackage{algorithm}
\usepackage{enumitem}
\usepackage{makecell}
\usepackage{algpseudocode}
\usepackage{booktabs} % For professional looking tables  
\usepackage{multirow} % For multirow cells  
\usepackage{adjustbox} % To adjust table size if needed 
\usepackage{graphicx}  
\usepackage{booktabs}  
\usepackage{multirow}  
\usepackage{xcolor}
\usepackage{subcaption}  
\usepackage{float} 
\usepackage{colortbl}
\definecolor{lightblue}{RGB}{235, 242, 255}  
\definecolor{lightyellow}{RGB}{255, 250, 220}  
\definecolor{lightgreen}{RGB}{230, 255, 230}  
\definecolor{lightpurple}{RGB}{250, 240, 255}  
\definecolor{mycolor}{RGB}{255, 245, 232}

\AtBeginDocument{%
  }

\copyrightyear{2025}
\acmYear{2025}
\setcopyright{acmlicensed}\acmConference[KDD '25]{Proceedings of the 31st ACM SIGKDD Conference on Knowledge Discovery and Data Mining V.2}{August 3--7, 2025}{Toronto, ON, Canada}
\acmBooktitle{Proceedings of the 31st ACM SIGKDD Conference on Knowledge Discovery and Data Mining V.2 (KDD '25), August 3--7, 2025, Toronto, ON, Canada}
\acmDOI{10.1145/3711896.3736985}
\acmISBN{979-8-4007-1454-2/2025/08}

\begin{document}

\title{GORACS: Group-level Optimal Transport-guided Coreset Selection for LLM-based Recommender Systems}

\author{Tiehua Mei}
\orcid{0009-0005-9677-4653}
\affiliation{%
    \department{School of Data Science}
    \institution{Fudan University}
    \city{Shanghai}
    \country{China}
}
\email{thmei24@m.fudan.edu.cn}

\author{Hengrui Chen}
\orcid{0009-0002-3310-9224}
\affiliation{%
    \department{School of Data Science}
    \institution{Fudan University}
    \city{Shanghai}
    \country{China}
}
\email{chenhr24@m.fudan.edu.cn}

\author{Peng Yu}
\orcid{0009-0005-0335-9152}
\affiliation{%
    \department{School of Data Science}
    \institution{Fudan University}
    \city{Shanghai}
    \country{China}
}
\email{pyu22@m.fudan.edu.cn}

\author{Jiaqing Liang}
\orcid{0000-0003-0670-5602}
\affiliation{%
    \department{School of Data Science}
    \institution{Fudan University}
    \city{Shanghai}
    \country{China}
}
\email{liangjiaqing@fudan.edu.cn}

\author{Deqing Yang}
\orcid{0000-0002-1390-3861}
\authornote{Corresponding author.}
\affiliation{%
    \department{School of Data Science}
    \institution{Fudan University}
    \city{Shanghai}
    \country{China}
}
\email{yangdeqing@fudan.edu.cn}

\renewcommand{\shortauthors}{Tiehua Mei, Hengrui Chen, Peng Yu, Jiaqing Liang, \& Deqing Yang}

\begin{abstract}
Although large language models (LLMs) have shown great potential in recommender systems, the prohibitive computational costs for fine-tuning LLMs on entire datasets hinder their successful deployment in real-world scenarios. To develop affordable and effective LLM-based recommender systems, we focus on the task of \emph{coreset selection} which identifies a small subset of fine-tuning data to optimize the test loss, thereby facilitating efficient LLMs' fine-tuning. Although there exist some intuitive solutions of subset selection, including distribution-based and importance-based approaches, they often lead to suboptimal performance due to the misalignment with downstream fine-tuning objectives or weak generalization ability caused by individual-level sample selection. To overcome these challenges, we propose \textbf{GORACS}, which is a novel \textbf{G}roup-level \textbf{O}ptimal t\textbf{RA}nsport-guided \textbf{C}oreset \textbf{S}election framework for LLM-based recommender systems. GORACS is designed based on two key principles for coreset selection: 1) selecting the subsets that \emph{minimize the test loss} to align with fine-tuning objectives, and 2) enhancing model generalization through \emph{group-level} data selection. Corresponding to these two principles, GORACS has two key components: 1) a Proxy Optimization Objective (POO) leveraging optimal transport and gradient information to bound the intractable test loss, thus reducing computational costs by avoiding repeated LLM retraining, and 2) a two-stage Initialization-Then-Refinement Algorithm (ITRA) for efficient group-level selection. Our extensive experiments across diverse recommendation datasets and tasks validate that GORACS significantly reduces fine-tuning costs of LLMs while achieving superior performance over the state-of-the-art baselines and full data training. The source code of GORACS are available at \url{https://github.com/Mithas-114/GORACS}.
\end{abstract}




\maketitle

%\newcommand\kddavailabilityurl{https://doi.org/10.5281/zenodo.15534886}

%\ifdefempty{\kddavailabilityurl}{}{
%\begingroup\small\noindent\raggedright\textbf{KDD Availability Link:}\\
%The source code of this paper has been made publicly available at \url{\kddavailabilityurl}.
%\endgroup
%}

\section{INTRODUCTION}
Large language models (LLMs) have demonstrated remarkable success in a wide range of recommendation tasks \cite{llmrec_intro_1, llmrec_intro_2, chat_rec} due to their vast knowledge and advanced capabilities \cite{recommender_era_llm}. These recommendation tasks can be mainly categorized into two paradigms \cite{wu2024survey}. The first is {\emph{discriminative recommendation}}, where LLMs predict recommendation results from a predefined label set, such as click-through rate (CTR) \cite{bao2023tallrec} or rating prediction \cite{llmrec_intro_1}. The second is {\emph{generative recommendation}}, where LLMs generate open-ended recommendation information for complex scenarios, such as sequential recommendation \cite{bao2023bistepgroundingparadigmlarge}, explanation generation \cite{llm_explain}, and conversational recommendation \cite{llm_conversation}.

In general, achieving the optimal performance of LLM-based recommender systems (LLMRecs) requires instruction fine-tuning LLMs on large-scale recommendation datasets \cite{hllm}, which often incurs unaffordable computational costs \cite{llm_rec_cost}. This challenge has made the development of efficient fine-tuning methods for LLM-based recommender systems a critical area of research. While existing parameter-efficient fine-tuning (PEFT) methods can reduce training costs by updating only a small subset of model parameters, this approach alone is insufficient to address the high computational demands posed by ever-growing recommendation datasets. In contrast, recent studies \cite{llm_coreset_related_1, LIMAzhou} in related domains have shown that fine-tuning LLMs on carefully selected small subsets can significantly reduce computational overheads while maintaining or even boosting model performance. It is an observation aligning with recent findings \cite{rsdatasurvey} in recommender systems which highlights the key role of data quality in improving both model performance and training efficiency. However, this promising data-side optimization strategy, commonly referred to as \emph{coreset selection}, remains seldom explored for LLM-based recommender systems.

The goal of coreset selection is to minimize the test loss by selecting a small but representative subset of whole training data with the given budget, thus enabling efficient fine-tuning \cite{llm_coreset_related_2}. However, existing techniques of coreset selection, including \emph{distribution-based methods} and \emph{importance-based methods}, often struggle to achieve this goal. Distribution-based methods \cite{D2, CCS, DSIR} aim to cover the entire dataset through stratified sampling or graph-based algorithms. While effective on capturing feature space distributions, these methods fail to directly minimize the test loss and suffer from poor alignment with the optimization objectives of downstream fine-tuning tasks, resulting in suboptimal performance \cite{BOSS}. On the other hand, importance-based methods \cite{dealrec, EL2N, importance_method} rank samples according to their training contribution and select top-$K$ samples. However, such individual-level selection strategy often overemphasizes the high-importance samples near decision boundaries, limiting the model's generalization to other samples \cite{if_not_good}. Moreover, in recommender systems, data characteristics like user-item interactions and temporal dependencies naturally form inter-sample correlations. However, individual-level methods \cite{EL2N} which focus on isolated samples, inherently overlook these collective structures, thus failing to create a truly representative coreset.

To address these limitations, we identify two key objectives for coreset selection task to improve LLMRecs: (O1) selecting the subsets that {\emph{minimize the test loss}} to align with downstream objectives; (O2) adopting {\emph{group-level}} subset selection, i.e., evaluating the collective quality of a group of samples together, rather than separately considering each individual sample’s importance, to capture inherent inter-sample correlations in the recommendation data and ensure the model's generalization capability. However, achieving these two objectives still faces two major challenges. 

    \noindent \textbf{Q1 - Computational overhead}: Computing the test loss for any subset is often prohibitive, as it requires retraining LLMs on each candidate subset, which is infeasible given the computation resource constraints of real applications \cite{trak}.
    
    \noindent \textbf{Q2 - Combinatorial explosion}: Group-level coreset selection inevitably involves searching across an exponentially large candidate space of groups, making the optimization greatly more complex than traditional individual-level importance-based methods \cite{hard-np}.

To overcome these challenges, in this paper we propose a novel \textbf{G}roup-level \textbf{O}ptimal t\textbf{RA}nsport-guided \textbf{C}oreset \textbf{S}election framework for LLMRecs, namely \textbf{GORACS}. Our framework consists of two key components corresponding to the challenges: a computationally efficient \emph{Proxy Optimization Objective (POO)} and a two-stage \emph{Initialization-Then-Refinement Algorithm (ITRA)}.

    \noindent \textbf{Proxy Optimization Objective (POO)}: To reduce the cost of computing the test loss (\textbf{Q1}), we develop a proxy objective POO combining optimal transport (OT) distance \cite{villani2003topics} and gradient information. Leveraging Kantorovich-Rubinstein duality \cite{KRD}, we bound the difference between training loss and test loss using the OT distance. Additionally, we bound training loss efficiently via gradient norm analysis, thus avoiding repeated model retraining and evaluation. By integrating these approaches, we derive the POO as an upper bound of the test loss. This enables us to estimate the test loss using the POO, leading to significantly reduced computational overhead in subset quality assessment.
    
    \noindent \textbf{Initialization-Then-Refinement Algorithm (ITRA)}: To tackle the combinatorial complexity of group-level subset selection (\textbf{Q2}), the first stage of ITRA solves a relaxed form of the proxy objective via greedy search to quickly generate a high-quality initial solution. The second stage refines this solution through sample exchanges, guided by a novel pruning strategy that identifies promising exchanges based on marginal improvement estimations. Our ITRA significantly reduces complexity of group-level optimization while ensuring the model's strong performance.

Furthermore, we extend our GORACS for discriminative recommendation tasks by incorporating label information. Specifically, we decompose the joint distribution into class-conditional components, enabling fine-grained selection for each class while maintaining balanced class proportions. Our extensive experiments conducted on both generative and discriminative recommendation tasks across multiple datasets validate GORACS' effectiveness. 

In summary, our major contributions in this paper include:

     1. We propose a group-level coreset selection framework GORACS based on optimal transport to address the challenge of selecting coresets for efficient LLMRecs fine-tuning. Our framework successfully bridges the gap between data selection and downstream task performance, effectively achieving test loss minimization.
    
     2. We design a novel proxy optimization objective (POO) to reduce computational overhead of subset quality assessment, and introduce an efficient two-stage ITRA algorithm to tackle the combinatorial explosion of group-level selection, thus enabling efficient and effective coreset selection.
    
     3. We further enhance GORACS for discriminative recommendation tasks by incorporating label information, which ensures fine-grained class representation and improves classification performance. Our extensive experiments across generative and discriminative tasks on multiple datasets validate the effectiveness and efficiency of GORACS and its components.

\section{RELATED WORK}
\subsection{LLM-based Recommendation}
LLMs have introduced new possibilities for recommender systems by leveraging their broad knowledge and advanced capabilities \cite{wu2024survey,fairnesssurvey}. Although methods such as in-context learning and prompting \cite{popularity_2,Isgptliu} have been explored, a major challenge in LLMRecs is aligning LLMs with the specific requirements of recommendation tasks. Recently, instruction fine-tuning has shown promise in improving LLMs' adaptability for recommendation \cite{ALLM,bao2023tallrec} while it is computationally expensive and highly dependent on high-quality data. Notably, high-quality data has been shown to outperform large-scale datasets on improving model performance \cite{wangdatasurvey,fromqtoqli}. Zhou et al. \cite{LIMAzhou} demonstrated that fine-tuning LLMs on as few as 1,000 carefully selected samples can significantly boost generalization to unseen tasks, underscoring the critical role of coreset selection, which however remains underexplored in the context of LLMRecs.

\subsection{Coreset Selection}
Existing coreset selection methods \cite{importance_method,ingenicorslect,moderate} can be broadly categorized into two types. 1) Distribution-based methods \cite{craig, CCS, D2} seek to select a subset that preserves the dataset distribution in feature space through various distribution matching or covering strategies. Wherein, FDMat \cite{fdmat} employs optimal transport for distribution matching, which is technically similar to our GORACS. However, these methods including FDMat, neglect directly optimizing test loss and lack aligning with downstream fine-tuning. 2) Importance-based methods \cite{forgetscore,farewell, EL2N, adacore} rank and select samples based on difficulty metrics, assuming that harder samples are more valuable for training. Since previous metrics are often computationally intensive \cite{LESS} and thus impractical for LLMRecs, DEALRec \cite{dealrec} leverages a surrogate recommender model to efficiently estimate the influence of removing individual samples on the training loss. Despite their contributions, these methods often prioritize high-impact individual samples that always locate on the decision boundary, thus hindering the model’s generalization capability to other samples \cite{if_not_good,CCS}. To address these issues, our GORACS directly optimizes for test loss and leverages group-level selection, effectively improving recommendation fine-tuning performance.

\section{PRELIMINARIES}
Before presenting our framework, we first introduce preliminaries on LLMRecs and the coreset selection tasks. We also cover key concepts of optimal transport which are the basics of our method.

\subsection{LLM-based Recommender Systems}
LLMRecs leverage LLMs to generate recommendation results by converting recommendation tasks into Q\&A problems. In general, the input data for the LLM, such as user-item historical interactions, are formatted as the prompt $\boldsymbol{x}$ to encourage the LLM to output the results $\boldsymbol{y}$. LLMRecs can be mainly categorized into \cite{wu2024survey}:
\begin{itemize}[leftmargin=*]
    \item Discriminative Recommendation: The LLM predicts (selects) the recommendation results from a small candidate (label) set, such as click-through rate (CTR) \cite{bao2023tallrec} or rating prediction \cite{llmrec_intro_1}.
    \item Generative Recommendation: The LLM directly generates open-ended recommendation results for complex scenarios, such as sequential recommendations \cite{bao2023bistepgroundingparadigmlarge} or explanation generation \cite{llm_explain}.
\end{itemize}
However, as LLMs lack specialized training on recommendation data, fine-tuning is essential to develop effective LLMRecs \cite{dealrec}, which optimizes parameters $\phi$ by minimizing the training loss:
\begin{equation}
    \label{LLMRec_loss}
    \min_{\phi} \left\{\mathcal{L}_\phi(\mathcal{T})=\frac{1}{|\mathcal{T}|}\sum_{(\boldsymbol{x}, \boldsymbol{y})\in \mathcal{T}}\sum_{t=1}^{|\boldsymbol{y}|}-\log P_\phi(\boldsymbol{y}_t|\boldsymbol{x}, \boldsymbol{y}_{<t})\right\},
\end{equation}
where $\mathcal{T}=\{(\boldsymbol{x}_i, \boldsymbol{y}_i)\}_{i=1}^{|\mathcal{T}|}$ represents the training (fine-tuning) dataset. $\boldsymbol{y}_t$ is the $t$-th token in the token sequence $\boldsymbol{y}$, and $\boldsymbol{y}_{<t}$ denotes the tokens before $\boldsymbol{y}_t$. However, fine-tuning on the entire dataset is generally expensive, making efficiency improvement crucial for developing LLMRecs \cite{llm_rec_cost}.

\subsection{Coreset Selection Task for LLM-based Recommendation}\label{sec:task}
To reduce training costs, recent studies have explored fine-tuning LLMRecs on the subsets of full training data. However, existing data selection strategies, such as random sampling \cite{bao2023tallrec} and influence-based methods \cite{dealrec}, do not directly optimize test performance of LLMRecs (i.e., {\emph{ minimizing test loss}}), leading to suboptimal results. To address it, we introduce the {\emph{coreset selection}} task for LLMRecs, which directly takes test loss minimization as the criterion for subset selection. Formally, consider a recommendation task with training dataset $\mathcal{T}$ containing $|\mathcal{T}|$ samples. The goal of coreset selection is to find the optimal subset $\mathcal{S}^*_{\text{opt}}$ of size $n$ from $\mathcal{T}$ that minimizes the expected loss over the test distribution (denoted by $\mathbb{P}$):
\begin{equation}
\label{Coreset_selection}
\mathcal{S}^*_{\text{opt}} = \underset{{\mathcal{S}:\mathcal{S}\subset \mathcal{T}}, |\mathcal{S}|= n}{\text{argmin}}\mathbb{E}_{\boldsymbol{z}\sim \mathbb{P}}[\mathcal{L}_{\phi^*_{\mathcal{S}}}(\boldsymbol{z})]\,\, \text{s.t.} \,\,\phi^*_{\mathcal{S}} = \underset{\phi}{\text{argmin}}\,\mathcal{L}_{\phi}(\mathcal{S}).
\end{equation}
Here, $\boldsymbol{z}=(\boldsymbol{x}, \boldsymbol{y})$ represents a recommendation data point. To the best of our knowledge, we are the first to apply the goal of Eq.\ref{Coreset_selection} in LLMRecs.  While bi-level optimization methods \cite{bi_level_1, glister, bi_level_2} have been utilized to solve Eq. \ref{Coreset_selection} in simpler scenarios, they are impractical for LLMRecs due to the high cost of training LLMs. Instead, we approach the solution of Eq. \ref{Coreset_selection} by analyzing the potential distributional gap between $\mathcal{S}$ and $\mathbb{P}$. Intuitively, if the distribution of subset $\mathcal{S}$ closely resembles $\mathbb{P}$, a model trained on $\mathcal{S}$ is likely to generalize well to $\mathbb{P}$, thereby achieving a lower expected test loss. To fulfill this insight, we employ the Optimal Transport distance \cite{villani2003topics} to effectively quantify the discrepancy between distributions.

\subsection{Basics on Optimal Transport}
\label{basic_ot}
Optimal Transport (OT) \cite{villani2003topics} is a mathematical theory for measuring the discrepancies between distributions, and we focus on its discrete version. Formally, let $(\mathcal{Z}, d)$ be a metric space with a metric $d:\mathcal{Z}\times \mathcal{Z}\to \mathbb{R}^+$. Suppose $\{\boldsymbol{z}_i\}_{i=1}^m\subset \mathcal{Z}$ and $\{\boldsymbol{z}'_j\}_{j=1}^n\subset \mathcal{Z}$. Then, given two discrete probability measures $\mu_1=\sum_{i=1}^m p_i\delta(\boldsymbol{z}_i), \mu_2=\sum_{j=1}^n q_j\delta(\boldsymbol{z}'_j)$ defined\footnote{Here $\delta(\cdot)$ denotes the Dirac delta function, and $\sum_i p_i=\sum_j q_j=1$.} on $\mathcal{Z}$ with probability mass vectors $\boldsymbol{p}=(p_i)_{i=1}^m, \boldsymbol{q}=(q_j)_{j=1}^n$, and a cost matrix $\boldsymbol{\mathrm{C}}\in \mathbb{R}^{m\times n}$, the OT distance between $\mu_1$ and $\mu_2$ with respect to $\boldsymbol{\mathrm{C}}$ is defined as
\begin{equation}
\label{OT_c}
    OT_{\boldsymbol{\mathrm{C}}}(\mu_1, \mu_2):=\min_{\boldsymbol{\pi}\in\Pi(\mu_1, \mu_2)}\langle\boldsymbol{\pi}, \boldsymbol{\mathrm{C}}\rangle_F,
\end{equation}
where $\Pi(\mu_1, \mu_2):=\{\boldsymbol{\pi}\in \mathbb{R}^{m\times n}:\sum_{i}\pi_{ij}=q_j$, $\sum_j\pi_{ij}=p_i, \pi_{ij}\ge 0\}$ denotes a collection of discrete distribution couplings between $\mu_1$ and $\mu_2$, and $\langle,\rangle_F$ represents the Frobenius inner product. Actually, Eq. \ref{OT_c} is a linear programming problem, for which many efficient computation methods have been proposed \cite{sinkhorn, compute_ot, ot_fast}. 
Additionally, $OT_{\boldsymbol{\mathrm{C}}}$ can also be derived from its dual problem \cite{villani2003topics}: 
$$OT_{\boldsymbol{\mathrm{C}}}(\mu_1, \mu_2)=\sup_{\boldsymbol{u} \oplus \boldsymbol{v}\le \boldsymbol{\mathrm{C}}} (\boldsymbol{p}^T\boldsymbol{u} + \boldsymbol{q}^T\boldsymbol{v}),$$ where $\boldsymbol{u}\oplus \boldsymbol{v}\le \boldsymbol{\mathrm{C}}$ denotes $u_i + v_j\le C_{ij}, \forall, (i, j)$. 
$\boldsymbol{u}\in\mathbb{R}^m$ and $\boldsymbol{v}\in\mathbb{R}^n$ are the dual variables of $OT_{\boldsymbol{\mathrm{C}}}$ associated with $\mu_1$ and $\mu_2$, respectively. 

\begin{figure*}[t!]  
    \centering  
    \includegraphics[width=\linewidth]{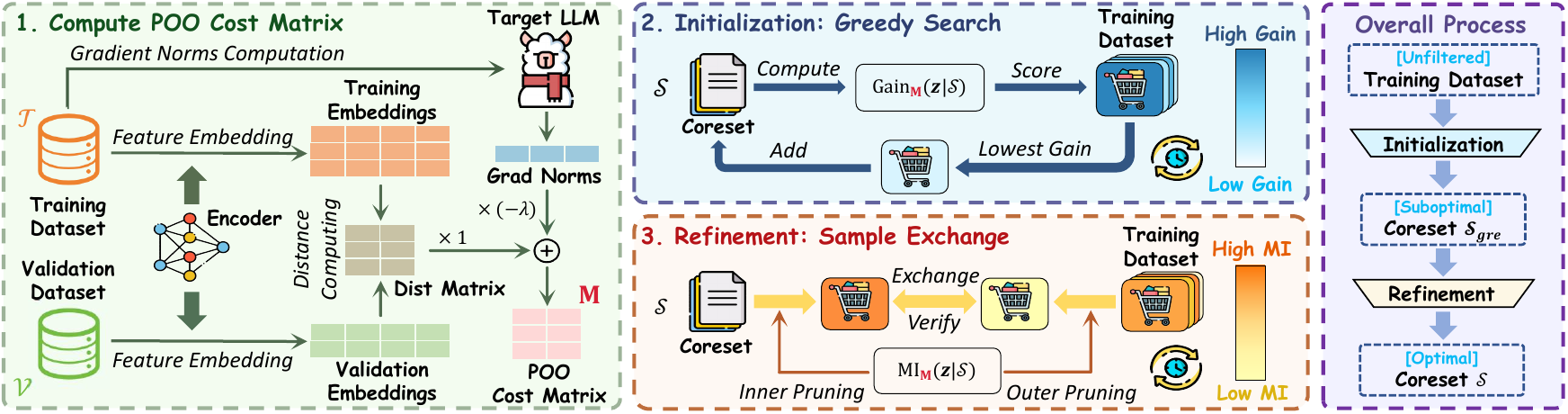}
    \caption{The overview of GORACS. It selects a representative coreset $\mathcal{S} \subset \mathcal{T}$ to minimize the POO score (Eq. \ref{score}) which is proven to be an upper bound on the test loss (Section \ref{sec:poo}). To this end, GORACS's pipline consists of three phases: 1) computing feature embedding distances and gradient norms to construct the POO cost matrix ${\boldsymbol{\mathrm{M}}}$ (Section \ref{sec:itra1}); 2) building an initial coreset $\mathcal{S}_{gre}$ by greedily adding samples with the lowest Gain scores (Section \ref{sec:itra2}); 3) refining $\mathcal{S}_{gre}$ via iteratively exchanging low-quality samples in $\mathcal{S}_{gre}$ with high-quality samples outside $\mathcal{S}_{gre}$, of which the quality is measured by the MI score (Section \ref{sec:itra3}).
    \vspace{-1em}
    }\label{fig:overview}
\end{figure*}

When using distance between points as the cost, i.e., using $\boldsymbol{\mathrm{D}}=(d(\boldsymbol{z}_i, \boldsymbol{z}'_j))_{ij}\in \mathbb{R}^{m\times n}$ as the cost matrix, the resulting $OT_{\boldsymbol{\mathrm{D}}}(\mu_1, \mu_2)$ enjoys a key advantage: it bounds the performance discrepancy of a model trained on one distribution and evaluated on another. This property is largely derived from the Kantorovich-Rubinstein Duality. Formally, let $\mathrm{Lip}-L$ denote the set of $L$-Lipschitz functions on $(\mathcal{Z}, d)$, i.e., $\mathrm{Lip}-L:=\{f:|f(\boldsymbol{z})-f(\boldsymbol{z}')|\le L\cdot d(\boldsymbol{z}, \boldsymbol{z}'), \forall \boldsymbol{z}, \boldsymbol{z}' \in \mathcal{Z}\}$. The Kantorovich-Rubinstein Duality \cite{KRD} states that
\begin{equation}
\label{KR}
    OT_{\boldsymbol{\mathrm{D}}}(\mu_1, \mu_2) = \frac{1}{L}\cdot \sup_{f\in \mathrm{Lip}-L} \left|\mathbb{E}_{\boldsymbol{z}\sim \mu_1}[f(\boldsymbol{z})]-\mathbb{E}_{\boldsymbol{z}'\sim \mu_2}[f(\boldsymbol{z}')]\right|.
\end{equation}
Thus,  a smaller $OT_{\boldsymbol{\mathrm{D}}}$ implies smaller difference between the expectations taken over two distributions.  When $f$ is chosen as a loss function, $\mathbb{E}_{\boldsymbol{z}\sim \mu}[f(\boldsymbol{z})]$ corresponds to the expected loss of the distribution $\mu$. This provides theoretical intuition for leveraging $OT_{\boldsymbol{\mathrm{D}}}$ as a proxy metric to evaluate the testing performance of a coreset.

\section{METHODOLOGY}
In this section, we present our coreset selection framework GORACS in detail. Specifically, we first design a proxy objective POO to approximate the solution of Eq. \ref{Coreset_selection}, and then propose the ITRA algorithm to solve the proxy optimization problem efficiently. Finally, we improve our framework for discriminative recommendation tasks by leveraging label information. The proofs of the theorems proposed in this section are presented in Appendix \ref{sec:appendix}. An overview of our approach is illustrated in Figure \ref{fig:overview}.

\subsection{Proxy Optimization Objective} 
\label{sec:poo}

As we mentioned before, directly optimizing Eq. \ref{Coreset_selection} is computationally infeasible due to the high cost of LLM fine-tuning. Therefore, we propose a Proxy Optimization Objective (POO) that tightly bounds the original criterion and remains computationally efficient. The POO consists of two components: 1) bounding the generalization gap between training loss and test loss using OT distance, and 2) bounding train loss via gradient norm analysis.

\subsubsection{{OT Distance Bounds for Recommendation Performance Gap}}\label{subsubsection:ot_bound}

As outlined in Section \ref{sec:task}, intuitively, when the discrepancy between the distribution of $\mathcal{S}$ and test distribution $\mathbb{P}$ is small, a model trained on $\mathcal{S}$ is likely to generalize well to $\mathbb{P}$, thereby reducing the gap between training loss and test loss. Building on Kantorovich-Rubinstein Duality (Eq. \ref{KR}), we leverage OT distance to quantify this generalization gap. To formalize this, let $\mu_{\mathcal{D}}:=\frac{1}{|\mathcal{D}|}\sum_{\boldsymbol{z}\in \mathcal{D}}\delta(\boldsymbol{z})$ denote the empirical distribution of a recommendation dataset $\mathcal{D}$. Each data point in $\mathcal{D}$, denoted by $\boldsymbol{z}$, represents a single instance containing user interaction information, which is formatted into a text prompt using recommendation-specific instruction templates. Following prior work \cite{get_more, tarot}, we embed $\boldsymbol{z}$ into $\mathbb{R}^N$ using a pre-trained encoder\footnote{We use Roberta-base\cite{liu2019roberta} to embed the textual data. We also compare different encoders in our ablation studies (Section \ref{sec:robustness}).} $E(\cdot)$. Given any metric\footnote{In this work, we simply utilize $L^2$ distance, while other metrics can also be applied.} $d$ on $\mathbb{R}^N$, we define the metric on $\mathcal{D}$ as $d^*(\boldsymbol{z}, \boldsymbol{z}')=d(E(\boldsymbol{z}), E(\boldsymbol{z}'))$, making $(\mathcal{D}, d^*)$ a metric space. Since the test distribution $\mathbb{P}$ is inaccessible, we follow established practice \cite{lava, glister} to approximate it using a held-out validation set $\mathcal{V}$. A more sophisticated strategy might involve exploiting temporal information within users' historical interactions to better simulate the test distribution\footnote{We leave this promising direction as future work.}. Thus, we propose the following theorem.
\begin{theorem}\label{lemma:1}
Let $\mathcal{D}$ be the full dataset, with training set $\mathcal{T}\subset \mathcal{D}$ and validation set $\mathcal{V} \subset \mathcal{D}$. Given a coreset $\mathcal{S}\subset \mathcal{T}$, suppose the loss function $\mathcal{L}_{\phi_{\mathcal{S}}^*}(\cdot)$ is $L$-Lipschitz with respect to the metric space $(\mathcal{D}, d^*)$. Denote $\mu_{\mathcal{S}}$ and $\mu_{\mathcal{V}}$ as the empirical distribution over $\mathcal{S}$ and $\mathcal{V}$ respectively. Let $\boldsymbol{\mathrm{D}}^*=(d^*(\boldsymbol{z}_i, \boldsymbol{z}'_j))_{ij}$ be the distance matrix between points in $\mathcal{T}$ and $\mathcal{V}$. Then the following inequality holds:
\begin{equation}
\mathbb{E}_{\boldsymbol{z}'\sim \mathbb{P}}[\mathcal{L}_{\phi_{\mathcal{S}}^*}(\boldsymbol{z}')] \leq \mathbb{E}_{\boldsymbol{z}\sim \mathbb{\mu_{\mathcal{S}}}}[\mathcal{L}_{\phi^*_{\mathcal{S}}}(\boldsymbol{z})] + L\cdot OT_{\boldsymbol{\mathrm{D}}^*}(\mu_{\mathcal{S}}, \mu_\mathcal{V}),
\label{eq:ot_bound}
\end{equation}
where $OT_{\boldsymbol{\mathrm{D}}^*}(\mu_{\mathcal{S}}, \mu_\mathcal{V})$ denotes the OT distance with cost matrix $\boldsymbol{\mathrm{D}}^*$. 
\end{theorem}
This bound includes two terms: 1) the training loss, which reflects the optimization dynamics of $\mathcal{S}$ but is costly for computation, and 2) the OT distance, which measures distributional discrepancy and is computationally efficient. Previous studies often assume that training loss is zero \cite{BOSS, CCS} for simplicity, yet LLMs in fine-tuning typically converge without reaching zero training loss, as final models reflect a combination of pre-training and fine-tuning distributions \cite{get_more, liu2019roberta}. Additionally, due to the differences in parameter sizes and pre-training data, LLMs possess distinct knowledge encoded in parameters, which impacts how LLMs utilize training samples to adapt to downstream tasks. To illustrate it, we quantify a training sample's contribution to training by its gradient norm, as it measures how much the sample updates model parameters and reflects the gap between the sample's information and the model's knowledge \cite{dealrec}. Figure \ref{games_grads} (a) shows distinct gradient norm distributions across LLMs on the Food dataset from Amazon, demonstrating models’ unique requirements for fine-tuning samples. Therefore, it is essential to preserve and efficiently estimate the training loss in Eq. \ref{eq:ot_bound} to effectively capture model-specific information.

\subsubsection{{Gradient-Based Analysis for Bounding Training Loss}} 
Inspired by the recent findings that early gradient norms effectively identify samples critical for training process \cite{EL2N}, we analyze how training data influence training via gradient descent, to estimate the training loss without fine-tuning on $\mathcal{S}$. Unlike prior work analyzing LLM training under stochastic gradient descent (SGD) \cite{craig}, we adopt full-batch gradient descent (GD) for theoretical analysis since we focus on training on small coresets (e.g., $|\mathcal{S}| \le 1024$). Therefore, the trainable parameters $\phi^t$ at step $t$ are updated as:
\begin{equation}  
    \label{gradient_descent}
    \phi^{t+1} = \phi^{t} -  \frac{ \eta^{t}}{|\mathcal{S}|}\sum_{\boldsymbol{z}\in \mathcal{S}} \nabla_\phi \mathcal{L}_{\phi^{t}}(\boldsymbol{z}).  
\end{equation}  
Focusing on the initial step $(t=0)$, we bound the training loss without full fine-tuning by proving the following theorem.
\begin{theorem} \label{theo:grad}
Consider the LLM fine-tuning following Eq. \ref{gradient_descent} on a small subset $\mathcal{S}\subset\mathcal{T}$ and suppose $H_\mathcal{S}(\phi)=\frac{1}{|\mathcal{S}|}\sum_{\boldsymbol{z}\in \mathcal{S}} \mathcal{L}_{\phi}(\boldsymbol{z})$ is $G$-smooth with respect to parameters $\phi$. If the learning rate $\eta^0$ at step $0$ satisfies $0<\eta^0<\frac{2}{G}$, then we have:
\begin{equation}
\label{theo4.2}
\begin{aligned}
    H_{\mathcal{S}}(\phi_{\mathcal{S}}^*)=\mathbb{E}_{\boldsymbol{z}\sim\mathbb{\mu_{\mathcal{S}}}}[\mathcal{L}_{\phi^*_{\mathcal{S}}}(\boldsymbol{z})] 
    \le \Lambda-\frac{C}{|\mathcal{S}|} \sum_{\boldsymbol{z}\in \mathcal{S}} \|\nabla_{\phi}\mathcal{L}_{\phi^0}(\boldsymbol{z})\|,
\end{aligned}
\end{equation}
where $\Lambda = \underset{\boldsymbol{z}\in \mathcal{T}}{\max} \,\mathcal{L}_{\phi^0}(\boldsymbol{z})$ and $C$ is a constant irrelevant to $\mathcal{S}$. 
\end{theorem}
Theorem \ref{theo:grad} indicates that the samples with larger initial gradient norms contribute more to training loss reduction. Moreover, we have conducted experiments with subsets\footnote{Note that the sizes of all subsets are the same 1,024.} of various average gradient norms to train BIGRec \cite{bao2023bistepgroundingparadigmlarge} on the dataset Games.
As illustrated in Figure \ref{games_grads} (b), there is a strong negative linear correlation ($R=-0.93$) between the normalized average gradient norms of a subset and its final training loss, which empirically supports Theorem \ref{theo:grad}.
\begin{figure}[t]  
  \centering    
  \includegraphics[width=\linewidth]{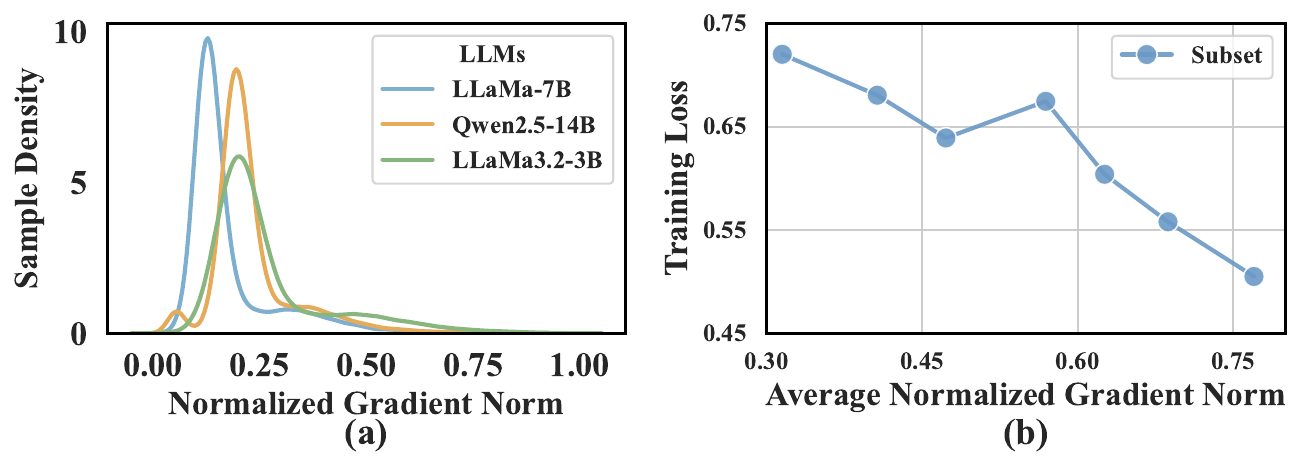}  
  \caption{(a) Distinct distributions of sample gradient norms of various LLMs. (b) The negative correlation between a subset's average gradient norms and its training loss.}  
  \label{games_grads}  
  \Description{This figure composes of two sub-figures. The left figure displays the different distributions of gradient norms of three LLMs, namely LLaMa-7B, Qwen2.5-14B and LLaMa3.2-3B. The right one shows a strong negative relationship between the gradient norms of subsets and their training losses.}
\end{figure} 

\subsubsection{{Overall Computationally Efficient Bound}}
By combining Theorem \ref{lemma:1} and Theorem \ref{theo:grad}, we derive the overall bound:
\begin{equation}
\label{bound}
\begin{aligned}
    \mathbb{E}_{\mathbb{P}}[\mathcal{L}_{\phi^*_{\mathcal{S}}}] \le L\cdot OT_{\boldsymbol{\mathrm{D}}^*}(\mu_{\mathcal{S}}, \mu_{\mathcal{V}}) - \frac{C}{|\mathcal{S}|}\sum_{\boldsymbol{z}\in \mathcal{S}}\|\nabla_{\phi}\mathcal{L}_{\phi^0}(\boldsymbol{z})\|+\Lambda.
\end{aligned}
\end{equation}

To select $\mathcal{S}$ that minimizes the test loss on the left-hand side, we can instead minimize the upper bound on the right-hand side. To this end, we define the following POO score (denoted by $\mathbb{S}(\cdot)$) to represent the expression on the right-hand side of Eq. \ref{bound}:
\begin{equation}
\label{score}
\mathbb{S}(\mathcal{S}) := OT_{\boldsymbol{\mathrm{D}}^*}(\mu_{\mathcal{S}}, \nu_\mathcal{V}) - \frac{\lambda}{|\mathcal{S}|} \sum_{\boldsymbol{z}\in \mathcal{S}}|\nabla_{\phi}\mathcal{L}_{\phi^0}(\boldsymbol{z})|,\,\,\, \textbf{[POO Score]}
\end{equation}
where $\lambda\ge 0$ is a hyper-parameter to balance the two terms. By minimizing $\mathbb{S}(\mathcal{S})$ with a proper $\lambda$, we identify an optimal subset $$ \mathcal{S}^*=\underset{{\mathcal{S}\subset\mathcal{T}, |\mathcal{S}|=n}}{\text{argmin}} \mathbb{S}(\mathcal{S}),$$  which ensures a low test loss as confirmed by Eq. \ref{bound}. Consequently, this \textbf{group-level} selection approach offers a practical and efficient method for approximating the optimal solution of Eq. \ref{Coreset_selection} using $\mathcal{S}^*$.

\subsection{Initialization-Then-Refinement Algorithm}

The Initialization-Then-Refinement Algorithm (ITRA) introduced in this part is developed to efficiently minimize the POO score $\mathbb{S}(\cdot)$ (Eq. \ref{score}). To this end, we first reformulate it as an OT distance with a special cost matrix $\boldsymbol{\mathrm{M}}$ (Eq. \ref{poo_matrix}) that combines embedding distance and gradient norm. Then, we propose a two-stage algorithm ITRA that fully utilizes the properties of OT distance. The first stage of ITRA 
employs constraint relaxation and greedy search to obtain an initial high-quality solution, and the second stage refines it via sample exchanges accelerated by a novel pruning strategy.

\subsubsection{{Reformulate POO}}
\label{sec:itra1}
Given $\mathcal{T}=\{\boldsymbol{z}_i\}_{i=1}^{|\mathcal{T}|}$ and $\mathcal{V}=\{\boldsymbol{z}_j'\}_{j=1}^{|\mathcal{V}|}$, the POO score $\mathbb{S}(\mathcal{S})$ can be equivalently expressed as the following OT distance, which directly follows from applying Eq. \ref{OT_c} to Eq. \ref{score}:
\begin{equation}
    \mathbb{S}(\mathcal{S}) = \min_{\boldsymbol{\pi}\in \Pi_{\mathcal{S}}}\langle\boldsymbol{\pi}, \boldsymbol{\mathrm{D}}^*-\lambda \boldsymbol{g}\cdot\boldsymbol{1}^T\rangle_F=OT_{\boldsymbol{\mathrm{M}}}(\mu_{\mathcal{S}}, \mu_\mathcal{V}). \label{score_ot}
\end{equation}
Here, $\boldsymbol{\mathrm{D}}^*=(d^*(\boldsymbol{z}_i, \boldsymbol{z}_j'))_{ij}\in\mathbb{R}^{|\mathcal{T}|\times|\mathcal{V}|}$ is the distance matrix (defined in Theorem \ref{lemma:1}), and $\boldsymbol{g}=(\|\nabla_{\phi}\mathcal{L}_{\phi^0}(\boldsymbol{z_i})\|)_i\in \mathbb{R}^{|\mathcal{T}|}$. In addition, $\Pi_\mathcal{S}:=\{\boldsymbol{\pi}\in \mathbb{R}^{|\mathcal{T}|\times |\mathcal{V}|}:\sum_{i}\pi_{ij}=\frac{1}{|\mathcal{V}|}, \sum_j\pi_{ij}=\frac{1}{|\mathcal{S}|} \mathbb{I}(\boldsymbol{z}_i\in \mathcal{S}), \pi_{ij}\ge 0\}$ is the coupling space, where $\mathbb{I}(\cdot)$ is the indicator function. Finally, the \textbf{POO cost matrix} $\boldsymbol{\mathrm{M}}$ for $OT_{\boldsymbol{\mathrm{M}}}$ is defined as
\begin{equation}
    \label{poo_matrix}
    \boldsymbol{\mathrm{M}}=\boldsymbol{\mathrm{D}}^*-\lambda \boldsymbol{g}\cdot\boldsymbol{1}^T=(D^*_{ij}-\lambda g_i)_{ij}\in\mathbb{R}^{|\mathcal{T}|\times|\mathcal{V}|}.
\end{equation}

Then, the proxy optimization objective can be reformulated into:
\begin{equation}
\label{opt:2}
    \mathcal{S}^*=\underset{{\mathcal{S}\subset\mathcal{T}, |\mathcal{S}|=n}}{\text{argmin}} \left(OT_{\boldsymbol{\mathrm{M}}}(\mu_{\mathcal{S}}, \mu_\mathcal{V})=\min_{\boldsymbol{\pi}\in \Pi_{\mathcal{S}}}\langle\boldsymbol{\pi}, \boldsymbol{\mathrm{M}}\rangle_F\right).
\end{equation} 

\subsubsection{{Relaxation and greedy search for initial solution}} 
\label{sec:itra2}
The bi-level structure of the optimization problem Eq. \ref{opt:2} poses a significant challenge, as the inner OT problem under constraint $\Pi_{\mathcal{S}}$ lacks a closed-form solution. However, we note that, by slightly relaxing the constraint space from $\Pi_{\mathcal{S}}$ to $\Omega_{\mathcal{S}}:=\{\boldsymbol{\pi}\in \mathbb{R}^{|\mathcal{T}|\times |\mathcal{V}|}:\pi_{ij}\ge 0, \sum_{i}\pi_{ij}=\frac{1}{|\mathcal{V}|},  \sum_j\pi_{ij}=0, \forall i\not \in \mathcal{S}\}$, the inner optimization over $\Omega_{\mathcal{S}}$ admits a closed-form solution: $\min_{\boldsymbol{\pi}\in \Omega_{\mathcal{S}}}\langle \boldsymbol{\pi}, \boldsymbol{\mathrm{M}}\rangle_F=\frac{1}{|\mathcal{V}|}\sum_{j=1}^{|\mathcal{V}|}\min_{\boldsymbol{z}_i\in \mathcal{S}}M_{ij}$.
Consequently, replacing $\Pi_{\mathcal{S}}$ with $\Omega_{\mathcal{S}}$ in Eq. \ref{opt:2} simplifies the bi-level optimization into the following p-median problem \cite{p_median}:
\begin{equation}
\label{p-median}
    \min_{\mathcal{S}\subset\mathcal{T}, |\mathcal{S}|=n} \left(\min_{\boldsymbol{\pi}\in \Omega_{\mathcal{S}}}\langle\boldsymbol{\pi}, \boldsymbol{\mathrm{M}}\rangle_F\right)\longleftrightarrow\min_{\mathcal{S}\subset\mathcal{T}, |\mathcal{S}|=n} \sum_{j=1}^{|\mathcal{V}|}\min_{\boldsymbol{z}_i\in \mathcal{S}} M_{ij},
\end{equation}
which enables a \textbf{greedy algorithm} \cite{p_median_greedy} to approximate the optimum solution of the problem Eq. \ref{p-median}. The greedy algorithm starts with an empty set $\mathcal{S}=\varnothing$ and keeps on adding data $\boldsymbol{z}\in \mathcal{T}\backslash \mathcal{S}$ to $\mathcal{S}$ that \textbf{minimizes the marginal gain}:
\begin{equation}
    \small
    \label{greedy}
    \text{Gain}_{\boldsymbol{\mathrm{M}}}(\boldsymbol{z}|\mathcal{S})=\sum_{j=1}^{|\mathcal{V}|}\min_{\boldsymbol{z}_i\in\mathcal{S}\cup \{\boldsymbol{z}\}} M_{ij}-\sum_{j=1}^{|\mathcal{V}|}\min_{\boldsymbol{z}_i\in\mathcal{S}} M_{ij}=\sum_{j=1}^{|\mathcal{V}|}(M_{\boldsymbol{z}j}-M_{*j})^-,
\end{equation}
where $x^-$ denotes $\min(x, 0)$, and $M_{*j}=\min_{\boldsymbol{z}_i\in\mathcal{S}} (M_{ij})$ needs to be computed only once per iteration. The solution obtained by this greedy algorithm is denoted by $\mathcal{S}_{\text{gre}}$. As an approximation of the optimal solution for the slightly relaxed problem Eq. \ref{p-median}, $\mathcal{S}_{\text{gre}}$ effectively minimizes a lower bound for the original optimization problem Eq. \ref{opt:2}, thus providing a strong insight for employing $\mathcal{S}_{\text{gre}}$ as an initial solution, which is further validated by our experimental results in Section \ref{sec:efficiency}.
\subsubsection{{Refinement via exchanges with pruning}} 
\label{sec:itra3}
To improve $\mathcal{S}_{\text{gre}}$, we employ an \textbf{exchange-based refinement} that repeatedly swaps elements between $\mathcal{S}$ and $\mathcal{T}\backslash\mathcal{S}$ \emph{whenever the swap leads to a decrease in $\mathbb{S}(\mathcal{S})$}. While adopted by combinatorial optimization \cite{exchange_1, exchange_2, exchange_3}, an exhaustive search requires at most $|\mathcal{S}|\times(|\mathcal{T}|-|\mathcal{S}|)$ OT distance calculations to identify beneficial exchanges, causing unaffordable cost for large recommendation datasets. Thus, we propose a pruning strategy that estimates the {\textbf{marginal improvement (MI)}} to identify potential exchanges, which is defined as:
\begin{equation}
    \text{MI}_{\boldsymbol{\mathrm{M}}}(\boldsymbol{z}|\mathcal{S}):=
    \begin{cases}
        \mathbb{S}(\mathcal{S}\cup\{\boldsymbol{z}\})-\mathbb{S}(\mathcal{S})\,\,\text{if}\,\,\boldsymbol{z}\not \in \mathcal{S}, \\
        \mathbb{S}(\mathcal{S})-\mathbb{S}(\mathcal{S}-\{\boldsymbol{z}\})\,\,\text{if}\,\,\boldsymbol{z}\in \mathcal{S}.
    \end{cases}
\end{equation}
%To this end,
Then, we use the dual of the OT distance and leverage its stability under small perturbations \cite{lava} to prove the following theorem.
\begin{theorem}
\label{theo:mi}
    $\mathrm{MI}$ score can be efficiently estimated as:
    \begin{align}
        \nonumber
        \mathrm{MI}_{\boldsymbol{\mathrm{M}}}(\boldsymbol{z}|\mathcal{S}) &\approx \sup_{y\in \mathbb{R}} F_{\boldsymbol{\mathrm{M}}}(y|\boldsymbol{z}, \mathcal{S}) \\
        F_{\boldsymbol{\mathrm{M}}}(y|\boldsymbol{z}, \mathcal{S}) &:= \frac{1}{|\mathcal{S}|}y+\frac{1}{|\mathcal{V}|}\sum_{j}(M_{\boldsymbol{z}j}-f^{\boldsymbol{\mathrm{M}}}_{\boldsymbol{z}j}(\boldsymbol{u}^*)-y)^-,\label{AI}
    \end{align}
    where $\boldsymbol{u}^*\in\mathbb{R}^{|\mathcal{T}|}$ denotes the optimal dual variables of $OT_{\boldsymbol{\mathrm{M}}}(\mu_{\mathcal{S}}, \mu_\mathcal{V})$ associated with $\mu_{\mathcal{S}}$ (defined in Section \ref{basic_ot}), and we define $f^{\boldsymbol{\mathrm{M}}}_{\boldsymbol{z}j}(\boldsymbol{u}^*)=\min_{\boldsymbol{z}_i\in \mathcal{S}:\boldsymbol{z}_i\not=\boldsymbol{z}} (M_{ij}-u^*_i).$
\end{theorem} 
Note that $F(y|\boldsymbol{z}, \mathcal{S})$ is piecewise linear with knots $M_{\boldsymbol{z}j}-f^{\boldsymbol{\mathrm{M}}}_{\boldsymbol{z}j}(\boldsymbol{u}^*), 1\le j\le |\mathcal{V}|$. Based on this fact, the optimal $\hat{y}_{\boldsymbol{z}}$ that maximizes $F_{\boldsymbol{\mathrm{M}}}(y|\boldsymbol{z}, \mathcal{S})$ equals to the $R$-th largest knot, where $R = \lceil |\mathcal{V}|/|\mathcal{S}|\rceil$. For efficient computation, we calculate $\boldsymbol{u}^*$ and $f^{\boldsymbol{\mathrm{M}}}_{\boldsymbol{z}j}(\boldsymbol{u}^*)$ once per iteration. For each candidate $\boldsymbol{z}$, we: 1) find the $R$-th largest value among $M_{\boldsymbol{z}j}-f^{\boldsymbol{\mathrm{M}}}_{\boldsymbol{z}j}(\boldsymbol{u}^*)$ as $\hat{y}_{\boldsymbol{z}}$, and 2) obtain MI$_{\boldsymbol{\mathrm{M}}}(\boldsymbol{z}|\mathcal{S})\approx F_{\boldsymbol{\mathrm{M}}}(\hat{y}_{\boldsymbol{z}}|\boldsymbol{z}, \mathcal{S})$. This process naturally supports parallel computation across candidates. The estimator MI$_{\boldsymbol{\mathrm{M}}}(\boldsymbol{z}|\mathcal{S}$) enables two efficient pruning strategies:

    \noindent 1. \emph{Outer pruning}: For the samples in $\mathcal{T}\backslash\mathcal{S}$, we rank them by $\text{MI}_{\boldsymbol{\mathrm{M}}}(\cdot|\mathcal{S})$ in ascending order and retain only top-$k$ candidates, as they have the highest potential for reducing $\mathbb{S}(\mathcal{S})$.
    
    \noindent 2. \emph{Inner pruning}: For a sample $\boldsymbol{z}\in\mathcal{S}$, a higher MI$_{\boldsymbol{\mathrm{M}}}(\boldsymbol{z}|\mathcal{S}$) indicates a greater potential reduction when removing $\boldsymbol{z}$, so we rank samples in descending order and select top-$k$ candidates.
    
By efficiently estimating MI scores and applying two pruning strategies, we greatly reduce the number of OT distance computations by \emph{only verifying the top-$k$ most promising candidate exchanges to search for a decrease in $\mathbb{S}(\mathcal{S})$}. If none of these candidates reduce $\mathbb{S}(\mathcal{S})$, the refinement terminates early.

\subsubsection{{Efficient OT computation}}\label{eff_ot_comp}
The value of $OT_{\boldsymbol{\mathrm{M}}}(\mu_{\mathcal{S}}, \mu_{\mathcal{V}})$ and the associated optimal dual variables $\boldsymbol{u}^*$ can be efficiently computed using Python's POT \cite{pot}. Notably, since $\mu_\mathcal{S}$ is supported only on $\mathcal{S}$, computation using the sub-matrix of $\boldsymbol{s}$ whose rows are indexed by $\mathcal{S}$ yields the equal OT value and the same optimal dual variables associated with $\mathcal{S}$. The dual variables on $\mathcal{T}\backslash\mathcal{S}$ are redundant and set to zero following \cite{lava}, thus significantly reducing the computational complexity from $\mathcal{|T|}\times \mathcal{|V|}$ to $|\mathcal{S}|\times |{\mathcal{V}}|$. The full procedure of our framework is detailed in Algorithm \ref{alg:1} in the Appendix \ref{app:alg}.

\subsubsection{Discussion} Algorithmically, GORACS achieves group-level selection by introducing the nonlinear OT distance in POO (Eq. \ref{score}) to capture inter-sample relationships. Although this design increases algorithmic complexity compared to individual-level methods, recommendation data naturally involve complex user-item interactions that form latent group connections within the data, making our OT-based group-level coreset selection framework particularly effective. Our ablation experiments in Section \ref{exp:ablation} confirm that the OT term is essential for capturing these structures and improving recommendation performance, clearly distinguishing our method from individual-level approaches in recommendation tasks.

\subsection{Label-enhanced Selection for Discriminative Recommendation}\label{sec:label}
We further enhance subset quality in classification tasks (e.g., discriminative recommendation) by incorporating label information into the subset selection process. The key insight is that in the classification task any joint distribution $\mathbb{Q}(\boldsymbol{x}, \boldsymbol{y})$ can be expressed as a weighted sum of class-conditional distributions $\mathbb{Q}(\boldsymbol{x}, \boldsymbol{y})=\sum_{k=1}^{K} q_k\cdot \mathbb{Q}_k(\boldsymbol{x})$, where $q_k$ is the class probability and $\mathbb{Q}_k$ is the conditional distribution for class $k$. We next show that this decomposition enables fine-grained selection for each class.

Let $\mathcal{V}_k$ denote the subset of validation samples with label $\boldsymbol{y}_k$, and $p_k=|\mathcal{V}_k|/|\mathcal{V}|$ be the class proportion. Then, for any subset $\mathcal{S}=\cup_{k=1}^K\mathcal{S}_k$ where $\mathcal{S}_k$ contains samples labeled $\boldsymbol{y}_k$ and satisfies $|\mathcal{S}_k|/|\mathcal{S}|=p_k$, we derive the following bound based on Theorem \ref{lemma:1} and Theorem \ref{theo:grad}:
\begin{equation}
\small
\nonumber
    \begin{aligned}
        &\mathbb{E}_{\mathbb{P}}[\mathcal{L}_{\phi_{\mathcal{S}}^*}(\boldsymbol{x}, \boldsymbol{y})]=\sum_{k=1}^K p_k\left(\mathbb{E}_{\mathbb{P_k}}[\mathcal{L}_{\phi_{\mathcal{S}}^*}^k(\boldsymbol{x})]- \mathbb{E}_{\mu_{\mathcal{S}_k}}[\mathcal{L}_{\phi_{\mathcal{S}}^*}^k(\boldsymbol{x})]\right) +\\
         & \mathbb{E}_{\mu_{\mathcal{S}}}[\mathcal{L}_{\phi_{\mathcal{S}}^*}(\boldsymbol{x}, \boldsymbol{y})]\le L\cdot \sum_{k=1}^K p_k \left(OT_{\boldsymbol{\mathrm{D}}^*}(\mu_{\mathcal{S}_k}, \mu_{\mathcal{V}_k})-\frac{\lambda}{|\mathcal{S}_k|} \sum_{\boldsymbol{z}\in \mathcal{S}_k} g_{\boldsymbol{z}}\right) + \Lambda,
    \end{aligned}
\end{equation}
where $\mathcal{L}^k(\boldsymbol{x})=\mathcal{L}(\boldsymbol{x}, \boldsymbol{y}_k)$, and $L, \lambda, \Lambda$ are constants. Unlike Eq. \ref{score}, this bound explicitly incorporates class-specific information, making it suitable for discriminative recommendations. Using the established Algorithm \ref{alg:1}, the bound can be optimized by independently minimizing $\mathbb{S}(\mathcal{S}_k, \mathcal{V}_k)=OT_{\boldsymbol{\mathrm{D}}^*}(\mu_{\mathcal{S}_k}, \mu_{\mathcal{V}_k})-\frac{\lambda}{|\mathcal{S}_k|}\cdot \sum_{\boldsymbol{z}\in \mathcal{S}_k} g_{\boldsymbol{z}}$ for each class under constraint $|\mathcal{S}_k|=p_k|\mathcal{S}|$, as detailed in Algorithm \ref{alg:2} in the Appendix \ref{app:alg}. Our experiments confirm that this label-aware approach significantly improves discriminative recommendations.

\section{EXPERIMENTS}\label{sec:exp}

\begin{table*}[t]  
    \centering  
    \small
    \caption{Overall performance comparison for SeqRec task. The best scores are highlighted in bold, while the second-best scores are \underline{underlined}. $\boldsymbol{\Delta\%}$ denotes the relative improvement percentage of our GORACS over the second-best competitors. }
    \label{combined_results}  
    \begin{tabular}{|l|ccccc|ccccc|ccccc|}  
        \Xhline{1.2pt}
        \multirow{2}{*}{\textbf{Methods}} & \multicolumn{5}{c|}{\textbf{Games}} & \multicolumn{5}{c|}{\textbf{Food}} & \multicolumn{5}{c|}{\textbf{Movies}} \\
        \cline{2-6} \cline{7-11} \cline{12-16}  
        & \textbf{TL}$\boldsymbol{\downarrow}$&\textbf{N@5} & \textbf{N@10} & \textbf{HR@5} & \textbf{HR@10} & \textbf{TL}$\boldsymbol{\downarrow}$&\textbf{N@5} & \textbf{N@10} & \textbf{HR@5} & \textbf{HR@10} &\textbf{TL}$\boldsymbol{\downarrow}$&\textbf{N@5} & \textbf{N@10} & \textbf{HR@5} & \textbf{HR@10} \\
       \Xhline{1.2pt}
        \textbf{Random} &0.8217 &0.1798 &0.2074& 0.2373 &0.3219& 0.8114 &0.0845 &0.1002 &\underline{0.1167} &0.1658 &0.8674 &\underline{0.1295} &0.1512&\underline{0.1717}&0.2392 \\
        \textbf{DSIR} & 0.8367& 0.1233 & 0.1494 & 0.1752 & 0.2572 & 0.9841&0.0705 & 0.0881 & 0.0997 & 0.1540 &1.1580 &0.0906 & 0.1137 & 0.1280 & 0.2002 \\
        \textbf{CCS} & 0.8467&\underline{0.1801} & \underline{0.2081} & \underline{0.2398} & \underline{0.3230} &0.8335 &0.0781 & 0.0944 & 0.1097 & 0.1607 & 0.9498&0.1285 & 0.1496 & 0.1708 & 0.2386 \\
        \textbf{D2} &0.8650 &0.1624 & 0.1888 & 0.2204 & 0.3020 &0.8140& 0.0720 & 0.0892 & 0.1057 & 0.1600 &0.9169& 0.1084 & 0.1321 & 0.1558 & 0.2296 \\
        \textbf{GraNd}&0.9815& 0.1546 & 0.1801 & 0.2020 & 0.2814 &1.0181& 0.0777 & 0.0959 & 0.1118 & \underline{0.1693} &1.2360& 0.0988 & 0.1226 & 0.1404 & 0.2152 \\ 
        \textbf{EL2N} &0.8367& 0.1182 & 0.1445 & 0.1632 & 0.2444 &1.0197& 0.0658 & 0.0824 & 0.0963 & 0.1478 &1.2380& 0.0837 & 0.1043 & 0.1214 & 0.1860 \\
        \textbf{DEALRec} &\underline{0.8214}& 0.1777 & 0.2046 & 0.2372 & 0.3208 &\underline{0.7923}& \underline{0.0851} & \underline{0.1016} & 0.1148 & 0.1665 &\underline{0.8443}& 0.1290 & \underline{0.1517} & 0.1706 & \underline{0.2414} \\
        \rowcolor{lightblue}  
        \textbf{GORACS} &\textbf{0.7650}& \textbf{0.1924} & \textbf{0.2195} & \textbf{0.2586} & \textbf{0.3404} &\textbf{0.7337}& \textbf{0.0910} & \textbf{0.1075} & \textbf{0.1236} & \textbf{0.1783} &\textbf{0.7643}& \textbf{0.1360} & \textbf{0.1610} & \textbf{0.1790} & \textbf{0.2568} \\
        \rowcolor{lightblue}  
         $\boldsymbol{\Delta\%}$ &\textbf{\textit{-6.87\%}} &\textbf{\textit{6.83\%}} & \textbf{\textit{5.48\%}} & \textbf{\textit{7.84\%}} & \textbf{\textit{5.39\%}} &\textbf{\textit{-7.34\%}}& \textbf{\textit{6.93\%}} & \textbf{\textit{5.81\%}} & \textbf{\textit{5.91\%}} & \textbf{\textit{5.32\%}} &\textbf{\textit{-9.48\%}}& \textbf{\textit{5.02\%}} & \textbf{\textit{6.13\%}} & \textbf{\textit{4.25\%}} & \textbf{\textit{6.38\%}} \\
        \Xhline{1.2pt} 
    \end{tabular}  
\end{table*}  

\begin{table}[t!]  
\centering  
\small  
\caption{Overall performance for CTRPre task. }  
\label{auc_testloss_results}  
\begin{tabular}{|l|cc|cc|cc|}  
\Xhline{1.2pt}  
\multirow{2}{*}{\textbf{Methods}} & \multicolumn{2}{c|}{\textbf{Games}} & \multicolumn{2}{c|}{\textbf{Food}} & \multicolumn{2}{c|}{\textbf{Movies}} \\
\cline{2-7}  
& \textbf{AUC}$\boldsymbol{\uparrow}$ & \textbf{TL}$\boldsymbol{\downarrow}$ & \textbf{AUC}$\boldsymbol{\uparrow}$ & \textbf{TL}$\boldsymbol{\downarrow}$  &   
 \textbf{AUC}$\boldsymbol{\uparrow}$ & \textbf{TL}$\boldsymbol{\downarrow}$ \\
\Xhline{1.2pt}   
\textbf{Random} & 0.5933 & 0.4903 & 0.5986 & 0.4837 & 0.6590 & \underline{0.4089} \\  
\textbf{DSIR} & 0.6278 & 0.4786 & 0.5664 & 0.5011 & 0.6565 & 0.4491 \\ 
\textbf{CCS} & 0.6381 & 0.4783 & \underline{0.6170} & 0.4864 & 0.6442 & 0.4868 \\  
\textbf{D2} & 0.6072 & 0.4915 & 0.5885 & 0.4874 & 0.5598 & 0.4769 \\  
\textbf{GraNd} & 0.4671 & 0.9954 & 0.4642 & 0.8897 & 0.4721 & 0.8663 \\  
\textbf{EL2N} & 0.4654 & 0.9953 & 0.4643 & 0.8907 & 0.4498 & 1.0032 \\   
\textbf{MODERATE} & 0.5385 & 0.5030 & 0.5533 & 0.4843 & \underline{0.6624}&0.4201\\ 
\textbf{FDMat} &\underline{0.6552} & \underline{0.4765} & 0.6099 &\underline{0.4836} & 0.6339& 0.4139\\
\rowcolor{lightblue}   
\textbf{GORACS} & \textbf{0.6949} & \textbf{0.4563} & \textbf{0.6306} & \textbf{0.4713} & \textbf{0.6944} & \textbf{0.3945} \\
\rowcolor{lightblue}   
$\Delta\%$ & \textbf{\textit{6.06\%}} & \textbf{\textit{-4.24\%}} & \textbf{\textit{2.20\%}} & \textbf{\textit{-2.54\%}} & \textbf{\textit{4.83\%}} & \textbf{\textit{-3.52\%}} \\  
\Xhline{1.2pt}  
\end{tabular}
\end{table}

\subsection{Experimental Settings}
\subsubsection{Dataset description} We conduct our experiments upon three widely used real-world datasets: Amazon \textbf{Games}, \textbf{Food} and \textbf{Movies}, all from the Amazon review datasets\footnote{\url{https://cseweb.ucsd.edu/~jmcauley/datasets/amazon_v2/}} which provide abundant user reviews and metadata. Table \ref{tab:datasets} summarizes the statistics of these datasets. We keep 5-core data for all datasets following \cite{ALLM, llamarec}, and sort user-item interactions chronologically to form interaction sequences. Each sequence contains a user's several consecutive historical item interactions as input and one subsequent item as output. We use the timestamp of the output item as the timestamp of the sequence. These sequences are then split chronologically into training, validation, and test sets to ensure no data leakage \cite{leakage}. Given the limitations in the inference speed of LLMs, we employ 8:1:1 split for the smaller Food dataset, while for larger Movies and Games we follow \cite{bao2023bistepgroundingparadigmlarge} and use the last 5,000 chronologically ordered sequences for test and the preceding 5,000 for validation.

\subsubsection{Tasks}
We evaluate GORACS on two key tasks in LLMRecs. 

     \noindent\textbf{1. Generative Sequential Recommendation (SeqRec):} This generative task requires LLMs to produce the next interacted item given a user's historical interaction sequence \cite{dealrec}. We adopt the competitive {\emph{BIGRec}} \cite{bao2023bistepgroundingparadigmlarge} as the backbone for its effectiveness and wide use in generative LLM-based  recommendation \cite{popularity_1, dealrec}. BIGRec represents items by generating item titles, and utilizes a $L^2$ embedding distance-based grounding paradigm to match generated item titles with the real item titles, thus ensuring accurate ranking.
     
     \noindent\textbf{2. CTR Prediction (CTRPre)}: This discriminative recommendation task classifies (predicts) target user's interaction as either ``like'' or ``dislike'' \cite{ALLM}, which has been extensively studied due to its effectiveness on shaping user decisions and improving personalized experiences \cite{ctr1, ctr2}. For this task, we adopt the representative {\emph{TALLRec}} \cite{bao2023tallrec} as the backbone, which predicts the target user's preference by outputting a binary label ``Yes'' or ``No'', based on the user's historical interacted items. Each item is represented by its title and labeled as ``like'' if the user's rating on it is greater than 3.

\vspace{-0.1cm}
\subsubsection{Baselines} {We compare GORACS with the following baselines of coreset selection. \textbf{Random} selects samples uniformly, which is a popular and strong baseline in coreset selection research \cite{deepcore}. 

{\emph{Distribution-based methods}}: \textbf{DSIR} \cite{DSIR} selects samples by aligning the n-gram frequencies of the selected coreset and the target distribution via importance resampling. \textbf{CCS} \cite{CCS} adopts an importance metric (we use EL2N following \cite{dealrec}) for stratified sampling to enhance data coverage in the coreset, which is competitive for low selection budgets. \textbf{D2 Pruning} \cite{D2} constructs graphs to update data scores and selects samples from diverse regions. 

{\emph{Importance-based methods}}: \textbf{GraNd} \cite{EL2N} selects important samples with higher gradient norms at early training stages. \textbf{EL2N} \cite{EL2N} selects the important samples whose prediction results are more different from the ground truth. \textbf{DEALRec} \cite{dealrec} is the state-of-the-art (SOTA) method designed for fine-tuning LLMRecs that identifies and selects influential samples by considering samples' influence scores and effort scores. Notably, DEALRec requires a small surrogate sequential recommendation model to compute influence scores, so we only compare it in the SeqRec task.} 

\begin{figure}[t]  
    \vspace{-0.3cm}
    \centering  
    \includegraphics[width=0.95\linewidth]{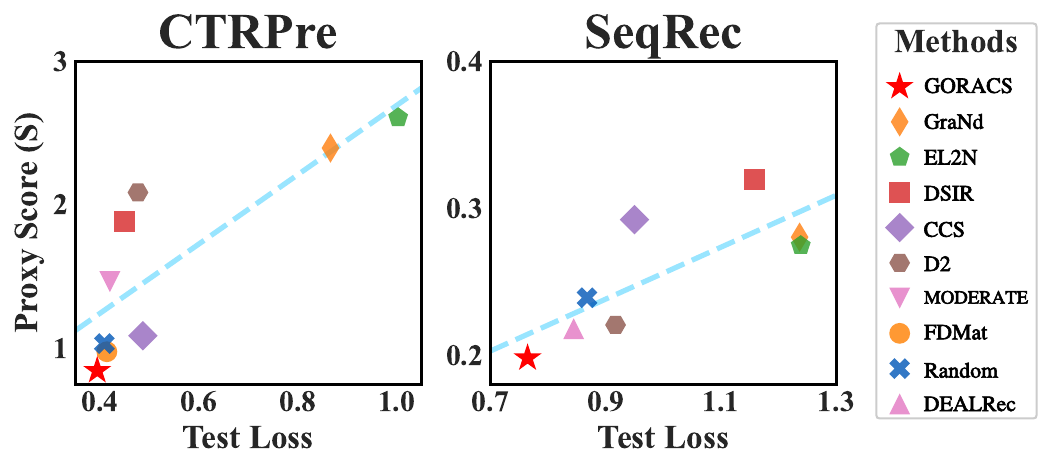} 
    \caption{Scatter Plots of Test Loss vs. Proxy Score on Movies when setting $\lambda$ of $\mathbb{S}$ to $0.1$ and $0.5$ respectively. The trend lines are derived from OLS regression analysis. }\label{fig:overall}  
    \Description{Scatter Plots of Test Loss vs. Proxy Score on Movies Dataset, which shows a linear relationship between the two terms. The $\lambda$ of $\mathbb{S}$ is set $0.1$ and $0.5$ respectively.}  
\end{figure}

For CTRPre task, we further add \textbf{MODERATE} \cite{moderate}, which selects samples at median distance from class center, and \textbf{FDMat} \cite{fdmat}, a class-aware method that uses optimal transport to select a coreset whose distribution matches the target distribution in the feature embedding space. See Appendix \ref{app:detail} for the implementation details.

\subsubsection{Evaluation metrics} For SeqRec task, we report the widely used metrics \textbf{HitRatio@$k$ (HR@$k$)} and \textbf{NDCG@$k$ (N@$k$)} \cite{bao2023bistepgroundingparadigmlarge, DROS}, where $k$ is set to $5/10$. Following \cite{ALLM, LLMRec} we randomly sample 99 items that the user has not previously interacted with as negative samples. For CTRPre task, we employ the representative \textbf{AUC} \cite{bao2023tallrec, auc1, auc2}. Moreover, we calculate \textbf{Test Loss (TL)} for both tasks to comprehensively evaluate fine-tuning performance. 

\subsection{Overall Performance}

\begin{figure}[t]  
        \centering  
        \includegraphics[width=1.02\linewidth]{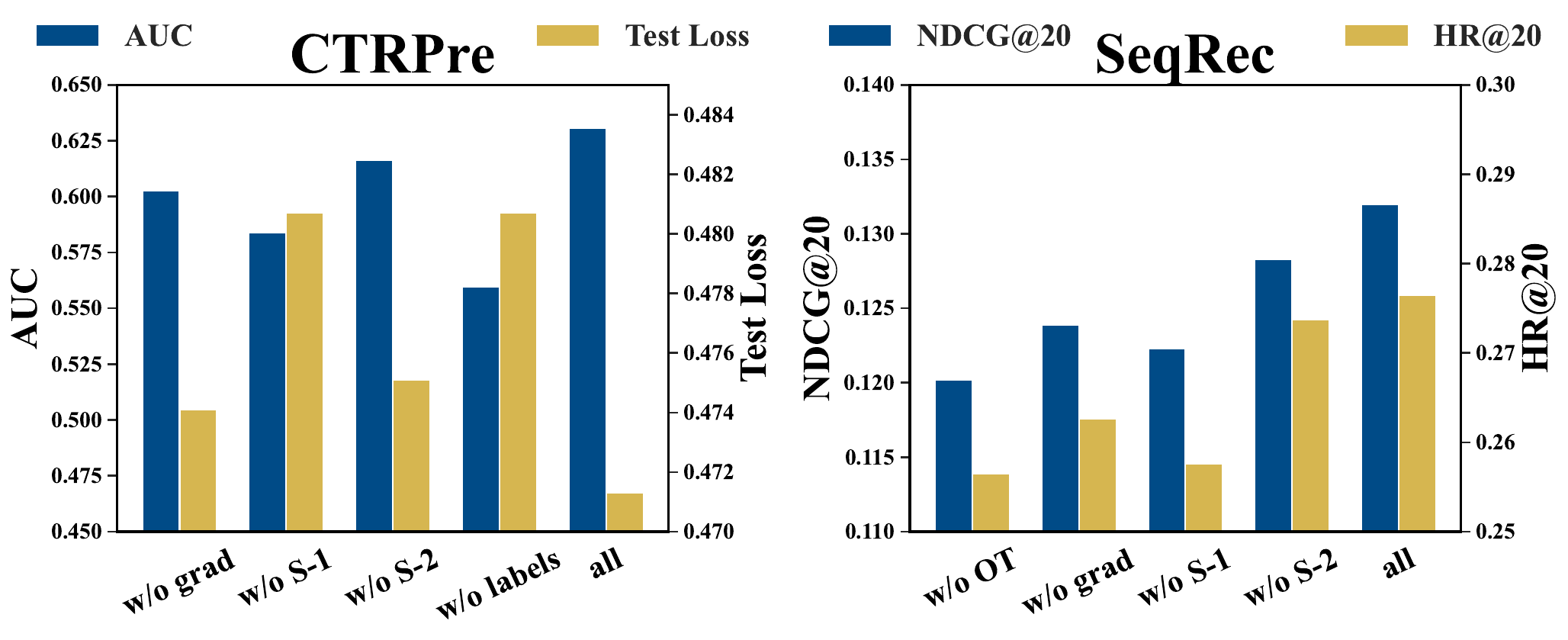} 
        \caption{Ablation studies of each component's contribution to the overall performance on Food. The “w/o OT” results on CTRPre (0.4721 for AUC and 0.8697 for Test Loss) were removed to improve figure presentation.}\label{fig:ablations}  
        \Description{Ablations that evaluate each component's contribution to the overall performance.} 
    \end{figure}

The performance scores of the baselines and GORACS on SeqRec and CTRPre task are presented in Table \ref{combined_results} and Table \ref{auc_testloss_results} respectively, from which we have the following observations and analysis.

1. Our proposed GORACS consistently outperforms all baselines for both SeqRec and CTRPre tasks on all datasets, justifying its robust generalization ability. Notably, GORACS consistently achieves the lowest Test Loss, highlighting its superior ability to imporve fine-tuning data by bridging the gap between coreset selection and downstream fine-tuning objectives. In contrast, while some methods (e.g., CCS, DEALRec, FDMat) achieve competitive results on certain datasets, none exhibits consistently strong performance across all settings. This inconsistency arises because the selection criteria of these methods do not directly align with the final fine-tuning objective, fundamentally limiting their generalization and stability compared to our approach.
    
2. All baselines exhibit notable performance disparities. Specifically, we observe that distribution-based methods like CCS and D2 generally outperform importance-based methods such as GraNd and EL2N. This deficiency arises since GraNd and EL2N prioritize difficult samples with high individual information, neglecting the essential role of other samples and resulting in a biased training subset \cite{CCS}. In contrast, CCS and D2 ensure balanced coverage of selected samples by collectively considering the overall diversity, demonstrating the effectiveness of group-level coreset selection.

 3. Although DEALRec achieves near-top NDCG@10 on Movies, its selection objective does not align directly with the fine-tuning loss, resulting in suboptimal performance. Additionally, DEALRec uses a heuristic weighted sum of influence and effort scores to measure each sample's importance, which may fail to capture the typically non-linear relationship of these two criteria in complex recommendation tasks \cite{weightsum, weightedsum2}. In contrast, GORACS optimizes a proxy objective that accurately bounds the loss and incorporates non-linear OT distance to effectively model complex relationships.
    
4. To validate the effectiveness of our proposed POO ($\mathbb{S}$), we present scatter plots of Test Loss versus $\mathbb{S}$ for both tasks on the Movies dataset in Figure \ref{fig:overall}. The results show that GORACS achieves the best optimization of $\mathbb{S}$ and, consequently, the lowest Test Loss. As depicted in the figure, DEALRec ranks second in both $\mathbb{S}$ and Test Loss, which is fairly consistent with its performance in Table \ref{combined_results}. The positive linear relationship between Test Loss and $\mathbb{S}$ further justifies the POO ($\mathbb{S}$) as an indicative objective for coreset selection. 

\subsection{In-depth Analysis}

    \begin{figure}[t!]  
        \centering  
        \includegraphics[width=1.05\linewidth]{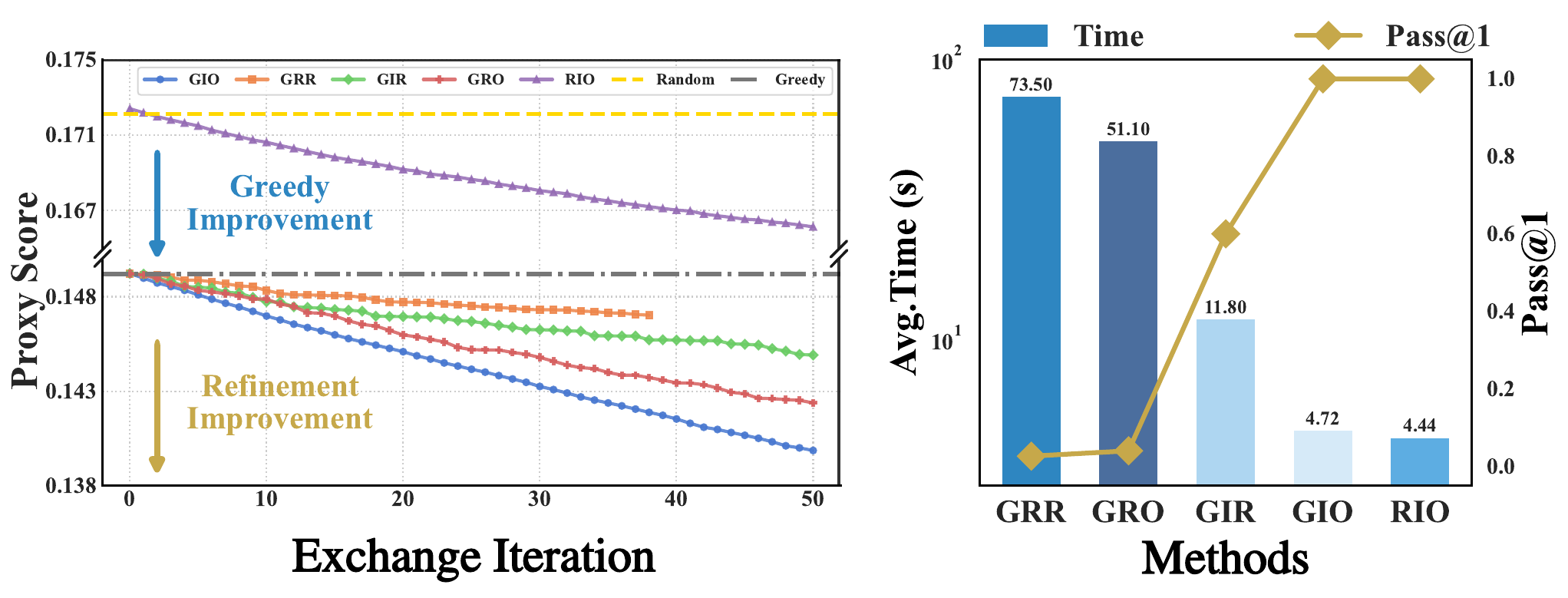} 
        \caption{Problem solving performance comparisons of ITRA variants on Games in terms of detailed optimization progress (left), and exchange time cost \& success ratio (right).} 
        \Description{}
        \label{fig:efficiency}  
    \end{figure} 
\subsubsection{{Ablation study}} \label{exp:ablation}
To assess the contributions of each component of GORACS, we conduct ablation studies by separately removing the OT distance term, the gradient norm term, the greedy search stage\footnote{In this case, we use randomly sampled subsets instead for initialization.} and the refinement stage, referred to as ``w/o OT'', ``w/o grad'', ``w/o S-1'' and ``w/o S-2'', respectively. We also replace Algorithm \ref{alg:2} with Algorithm \ref{alg:1} on CTRPre, termed as ``w/o labels'' to justify the impacts of incorporating label informantion.  
The results on Food are presented in Figure \ref{fig:ablations}, from which we observe that: 1) Removing OT distance or gradient norms degrades performance, while OT distance has a greater impact due to its essential role in measuring distribution discrepancies and capturing inter-sample relationships on group level. 2) Both the greedy search and refinement stage are critical for achieving high-quality solutions that better minimize test loss. 3) Neglecting label information on CTRPre significantly reduces GORACS's performance, highlighting labels' importance in capturing fine-grained class characteristics in discriminative tasks. In summary, GORACS's superior performance derives from its synergistic design that effectively integrates different components to address complex coreset selection task. 

\subsubsection{{Analysis of ITRA}}\label{sec:efficiency}

To assess the effectiveness and efficiency of the proposed ITRA algorithm and its components, i.e., greedy initialization (\textbf{G}), inner pruning (\textbf{I}) and outer pruning (\textbf{O}), we replace each with Random (\textbf{R}) and compare various combinations (e.g., \textbf{RIO},\textbf{GRR}, \textbf{GIR}, \textbf{GRO}, \textbf{GIO}) in terms of optimization process (Figure \ref{fig:efficiency} (left)) and exchange performance\footnote{Specifically, we compute two representative metrics Pass@1 (ratio of accepting the first candidate exchange) and Avg.Time (average time per successful exchange).} (Figure \ref{fig:efficiency} (right)). From the figure we observe that: 1) Greedy initialization provides a strong starting solution (0.149), significantly outperforming Random initialization (0.172), demonstrating its importance in setting a solid foundation. 2) Inner and outer pruning are critical for improving optimization performance and efficiency. The variants without them (e.g., GIR and GRO) perform poorly, and GRR even terminates prematurely due to rejecting all randomly searched candidate exchanges. 
Combining both strategies, GIO (i.e., ITRA) achieves the superior performance in terms of both effectiveness (fastest descent speed) and efficiency (perfect Pass@1 (100\%) and low average time cost). Notably, RIO has a slightly lower average time cost than GIO due to the absence of the greedy initiation stage in RIO, which only costs 19 seconds on the Games dataset with about 140,000 training samples. 4) As shown in Figure \ref{fig:efficiency} (right), inner pruning has the most significant impact on exchange success and time cost, likely due to its essential role in identifying the suboptimal samples mistakenly included in the early stage of greedy initialization. 

\subsubsection{Computational Efficiency}\label{sec:computational_complexity}
To further evaluate the computational efficiency of GORACS, we conduct experiments on the SeqRec task with the Games dataset, comparing GORACS with DEALRec and full-data training. As shown in Table \ref{tab:computational_complexity}, we report recommendation metrics, the time costs for data selection and training, and the total flos\footnote{The total number of floating-point operations for the entire process, including both data selection and training.}. Notably, GORACS achieves superior recommendation performance with only 20\% of the total time consumption and 15\% of the total flos required by full-data training, demonstrating substantial gains in both effectiveness and efficiency. Meanwhile, both DEALRec and GORACS outperform full-data training, highlighting the practical benefits of coreset selection in efficient training of LLMRecs, which is consistent with prior findings \cite{dealrec, condensation}.
\begin{table}[t]  
    \centering  
    \caption{Computational cost comparison on Games. Select.T and Train.T represent time cost for data selection and training (measured in hours). Flos denotes the total floating point operations consumed in the entire process.}  
    
    \begin{tabular}{l|ccccc}  
        \Xhline{1.2pt}   
        \textbf{Methods} & \textbf{N@5}$\boldsymbol{\uparrow}$ & \textbf{H@5}$\boldsymbol{\uparrow}$ & \textbf{Select.T}$\boldsymbol{\downarrow}$&\textbf{Train.T}$\boldsymbol{\downarrow}$ & \textbf{Flos}$\boldsymbol{\downarrow}$\\
        \hline
        \textbf{DEALRec} & 0.1777 & 0.2372 & 1.75 & 1.34 & 1.07e+18 \\
        \textbf{GORACS} & 0.1924 & 0.2586 & 1.63 & 1.29 & 1.01e+18 \\
        \textbf{Full Data} & 0.1702 & 0.2302 & - & 14.2 & 6.73e+18 \\ 
       \Xhline{1.2pt} 
    \end{tabular}  
    \label{tab:computational_complexity}  
\end{table} 

\subsubsection{{Robustness across different embedding models and LLM backbones}}   \label{sec:robustness}

\begin{figure}[t!]  
    \begin{minipage}[c]{0.48\linewidth}  
        \centering  
        \captionsetup{type=table}  
        \setlength{\tabcolsep}{2pt}  
         \caption{\footnotesize{GORACS' Performance SeqRec of Games with different encoder models to compute OT distance.}}  
        \label{tab:encoder}  
        \begin{tabular}{l|ccc}  
            \toprule  
            \textbf{Enc.} & \textbf{TL}$\boldsymbol{\downarrow}$ & \textbf{N@10} & \textbf{HR@10} \\
            \midrule  
            \midrule  
            \textbf{Be.B} & 0.7652 & 0.2131 & 0.3272 \\
            \textbf{Ro.B} & 0.7650 & 0.2195 & 0.3404 \\
            \textbf{Ro.L} & 0.7604 & 0.2199 & 0.3418 \\
            \textbf{BGE} & 0.7545 & 0.2286 & 0.3512 \\
            \bottomrule  
        \end{tabular}  
    \end{minipage}  
    \hfill  
    \begin{minipage}[c]{0.5\linewidth}  
        \centering
        \includegraphics[width=\linewidth]{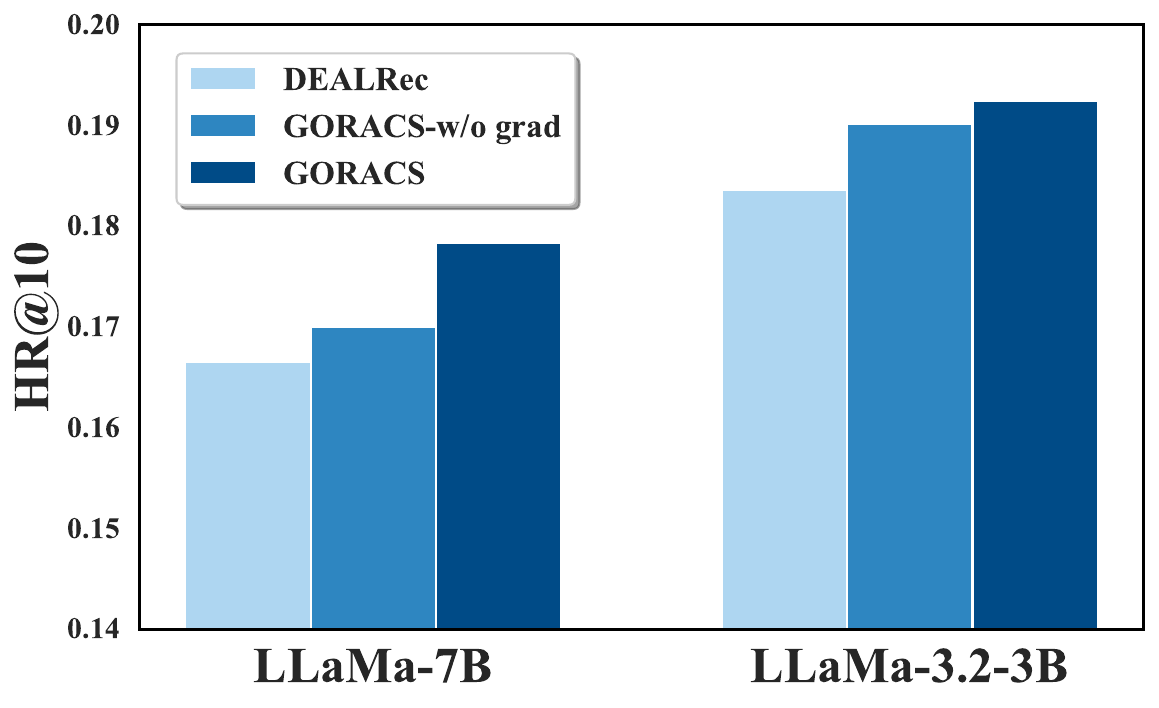}  
        \caption{\footnotesize{HR@10 on Food for DEALRec, GORACS (w/o grad) and GORACS applied to different LLMs.}}  
        \Description{The figure presents the SeqRec performance comparison of DEALRec, GORACS (without gradient information) and GORACS equipped with different LLMs backbone. The figure mainly displays the performance on Food dataset.}
        \label{fig:llm}  
    \end{minipage}  
    \vfill % Add vertical space to equalize height  
\end{figure}  

To evaluate the robustness of GORACS across diverse embedding models and LLM backbones, we employ four representative encoders, i.e., Bert-base (\textbf{Be.B}) \cite{bert}, RoBERTa-base (\textbf{Ro.B}) \cite{liu2019roberta}, RoBERTa-large (\textbf{Ro.L}) \cite{liu2019roberta}, and BGE-large-en-v1.5 (\textbf{BGE}) \cite{bge_embedding}, as the embedding models. As shown in Table \ref{tab:encoder}, stronger encoders consistently enhance GORACS’s performance by capturing recommendation-relevant features more precisely, enabling OT distance to better measure distributional discrepancies. However, the performance differences across different encoders remain small, demonstrating GORACS’s robustness on embedding quality. 
For backbone evaluation, we compare DEALRec, GORACS (w/o grad), and GORACS on SeqRec using LLaMA-7B \cite{llama} and LLaMA-3.2-3B-Instruct\cite{llama3.2}. The results in Figure \ref{fig:llm} indicate that GORACS, even without gradient information, consistently outperforms DEALRec, while incorporating gradient knowledge further improves its performance by leveraging model-specific information.

\subsubsection{{Impacts of coreset selection}}
We explore how GORACS enhances recommendation performance by selecting small, high-quality coresets over full-data training. Inspired by \cite{popularity_1, popularity_2}, we hypothesize that full-data training introduces popularity bias, as LLMs tend to memorize frequent popular items instead of capturing user preferences. To verify this, we fine-tune BIGRec with the selection budget $n$ from 64 to 2,048, plus the full dataset. We report HR@10 and Average Recommendation Popularity (ARP@10) \cite{ARP} to evaluate accuracy and popularity bias respectively. As shown in Figure \ref{fig:varying_k}, GORACS's recommendation performance often improves as $n$ increases, even surpassing the full-data trained model when $n\ge 256$, consistent with Section \ref{sec:computational_complexity}. Notably, popularity bias (ARP@10) first decreases as $n$ increases but rises again with full-data training. This occurs because very small coresets (e.g., $n=64$) are especially sensitive to the inclusion of popular items—just a few can dominate training and raise popularity bias. With larger selection budgets, GORACS can better balance popular and long-tail items, reducing popularity bias. However, in the full dataset, the abundance of popular items leads to memorization-driven overfitting and increases popularity bias again. Overall, the link between lower popularity bias and better recommendation performance confirms that popularity bias amplified by full-data training could harm recommendation quality.

\begin{figure}[t!]  
    \centering  
    \includegraphics[width=\linewidth]{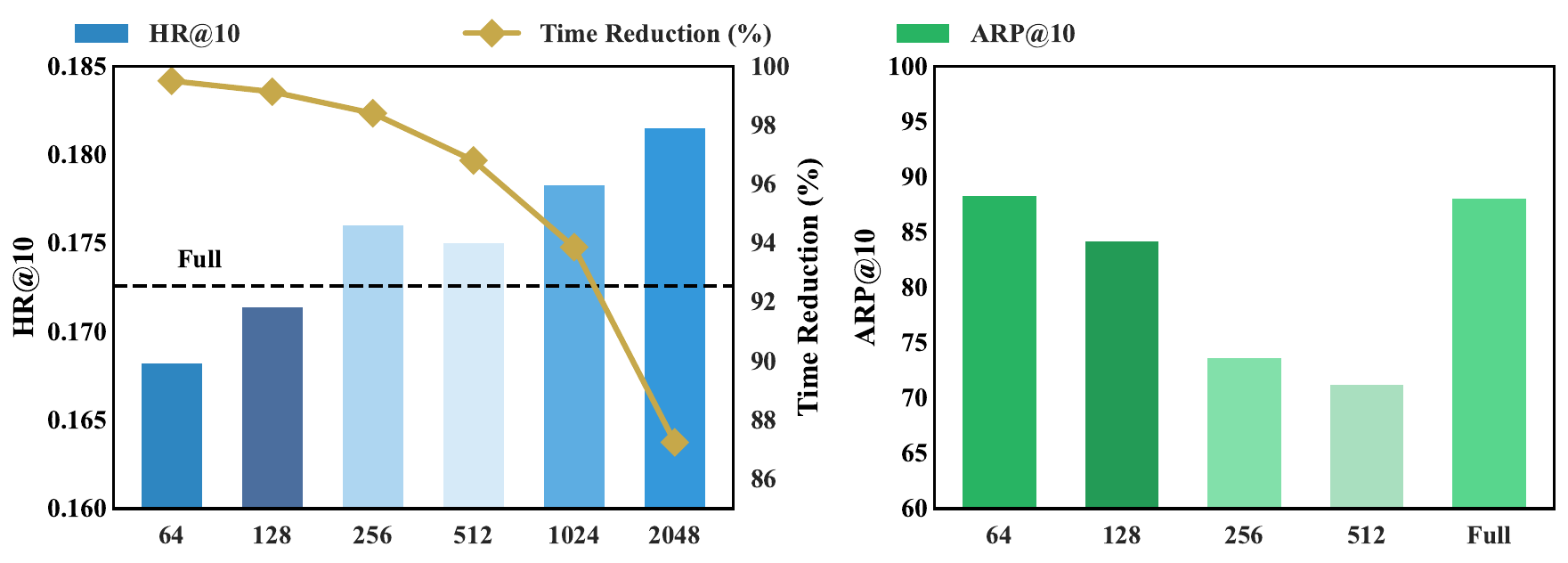}
    \caption{GORACS's performance (HR@10) of varying selection budgets, time reduction rate (compared to full dataset training) and popularity bias (ARP@10). }  
    \Description{This figure displays the performance of GORACS with varying selection budgets in terms of accuracy, time and popularity bias.}
    \label{fig:varying_k}  
\end{figure}

\section{CONCLUSION}
In this paper, we propose GORACS, a novel coreset selection framework for LLM-based recommender systems. GORACS introduces a proxy optimization objective (POO) leveraging optimal transport distance and gradient-based analysis, along with a two-stage algorithm (ITRA) for efficient subset selection. Our extensive experiments on two representative recommendation tasks verify that GORACS achieves SOTA performance and outperforms full dataset training while significantly reducing fine-tuning costs. By aligning coreset selection with downstream task objectives, GORACS provides a scalable and effective solution for applying LLMs to large-scale recommender systems. In the future work, we will explore applying GORACS to more complex recommendation tasks to further validate and extend its potential.

\begin{acks}
This work was supported by the Chinese NSF Major Research Plan (No.92270121).
\end{acks}

\clearpage
\bibliographystyle{ACM-Reference-Format.bst}
\bibliography{references.bib}

\appendix
\vspace{-1em}
\section{Appendix}
\subsection{Datasets and implementation details} \label{app:detail}

We conduct all experiments on four NVIDIA RTX A800 GPUs. For all the baselines and backbones, we use their open-source codes and follow the original settings in their papers. For BIGRec and TALLRec, We employ LLaMa-7B \cite{llama} with LoRA \cite{lora} for parameter-efficient fine-tuning, and set the selection budgets to 1,024 and 64 respectively, consistent with their original experimental settings. For GORACS, we search $\lambda$ in $\{0, 0.05, 0.1, 0.3, 0.5\}$. We apply our frameworks specified in Algorithm \ref{alg:1} and Algorithm \ref{alg:2} to SeqRec and CTRPre, respectively. For DEALRec, we utilize SASRec \cite{sasrec} to compute influence scores and search the regularization strength in $\{0.1, 0.3, 0.5, 0.7, 0.9\}$. We compute GraNd and EL2N using LLMs trained on the entire datasets for one epoch, as recommended in \cite{epoch}. For CCS, D2, and DSIR, we explore the number of strata, nearest neighbors, and hashed buckets in $\{25, 50, 75\}$, $\{5, 10, 20\}$ and $\{1000, 5000, 10000\}$, respectively. To ensure fairness, all embedding-based methods adopt the same RoBERTa-base \cite{liu2019roberta} encoder. All the optimal parameters are selected based on validation performance.
\begin{table}[h]  
    \centering  
    \caption{Statistics of datasets.} 
    \begin{tabular}{lcccc}  
        \toprule  
        \textbf{Datasets}& \textbf{\#Users} & \textbf{\#Items} & \textbf{\#Interactions} & \textbf{\#Sequences} \\
        \midrule
        \textbf{Games} &55,223 & 17,408 & 497,577 & 149,796\\ 
        \textbf{Food} &14,681 & 8,713 & 151,254 & 43,293\\ 
        \textbf{Movies}& 297,529 & 60,175 & 3,410,019 & 114,594\\
        \bottomrule  
    \end{tabular}  
    \label{tab:datasets}  
\end{table} 

\subsection{Scalability of GORACS}

To evaluate the scalability of GORACS, we conduct experiments on SeqRec task with the much larger MovieLens-1M dataset\footnote{\url{https://grouplens.org/datasets/movielens/}}(ML-1M), which contains about 930k sequences. Following Section \ref{sec:exp}, we fix the coreset size to 1,024 and the validation set size and test set size to 5k. We compare GORACS with Random Selection and DEALRec. As shown in Table \ref{appendix_scalability}, GORACS consistently outperforms the baselines across all recommendation metrics, demonstrating its effectiveness when applied to a larger dataset. Regarding efficiency, both DEALRec and GORACS spend significantly more time on coreset selection than on model training, since selection requires computing gradient norms (i.e., effort scores for DEALRec) over the entire training set. However, DEALRec’s selection time is longer due to the extra need to train a surrogate recommendation model. Importantly, GORACS’s coreset selection time scales nearly linearly with the dataset size: selection on ML-1M takes approximately 6.0 (9.8/1.63$\approx$6.0, see Section \ref{sec:computational_complexity}) times longer than on the Games dataset (150k sequences), closely matching their size ratio (930k/150k$\approx$6.6). This confirms the scalability of GORACS.

\begin{table}[h]  
    \centering  
    \caption{Performance comparison for SeqRec task on the larger MovieLens-1M dataset.}
    \label{appendix_scalability}  
    \begin{tabular}{l|cccc}  
        \Xhline{1.2pt}
        \multirow{2}{*}{\textbf{Methods}} & \multicolumn{4}{c}{\textbf{MovieLens-1M}}\\
        \cline{2-5}
        & \textbf{N@5}$\boldsymbol{\uparrow}$ & \textbf{H@5}$\boldsymbol{\uparrow}$ & \textbf{Select.T}$\boldsymbol{\downarrow}$&\textbf{Train.T}$\boldsymbol{\downarrow}$\\
       \Xhline{1.2pt}
        \textbf{Random} &0.1141 &0.1680 &-& 1.09\\
        \textbf{DEALRec} & 0.1178& 0.1720 & 11.7 & 1.18 \\
        \textbf{GORACS} & 0.1227&0.1806 & 9.8 & 1.11 \\
        \Xhline{1.2pt} 
    \end{tabular}  
\end{table}

\subsection{Performance of GORACS with Mistral-7B}

To further demonstrate that our framework generalizes to different LLM architectures beyond the LLaMA series evaluated in Section \ref{sec:robustness}, we conduct experiments using Mistral-7B-v0.3 \cite{mistral}. Specifically, we compare the SeqRec performance of GORACS against other baselines on the Games dataset with Mistral-7B-v0.3 as the backend model. As shown in Table \ref{appendix_more_llm}, GORACS consistently outperforms Random Selection and DEALRec across all recommendation metrics. These results indicate that GORACS effectively improves the quality of the coreset used to fine-tune Mistral, confirming its robustness and generalizability across different LLM architectures.

\begin{table}[h]  
    \centering  
    \caption{Performance comparison for the SeqRec task on the Games dataset using the backend model Mistral-7B-v0.3.}
    \label{appendix_more_llm}  
    \begin{tabular}{l|ccccc}  
        \Xhline{1.2pt}
        \multirow{2}{*}{\textbf{Methods}} & \multicolumn{5}{c}{\textbf{Games}}\\
        \cline{2-6}
        & \textbf{TL}$\boldsymbol{\downarrow}$ & \textbf{N@5}$\boldsymbol{\uparrow}$ & \textbf{N@10}$\boldsymbol{\uparrow}$ & \textbf{H@5}$\boldsymbol{\uparrow}$ & \textbf{H@10}$\boldsymbol{\uparrow}$\\
       \Xhline{1.2pt}
        \textbf{Random} & 0.9097& 0.1662 &0.1948 &0.2228& 0.3120\\
        \textbf{DEALRec} & 0.8916& 0.1719& 0.1990 & 0.2328& 0.3170 \\
        \textbf{GORACS} & 0.8113& 0.1835&0.2129 & 0.2492 & 0.3402 \\
        \Xhline{1.2pt} 
    \end{tabular}  
\end{table}

\subsection{Proofs of Theorems}\label{sec:appendix}
%\subsection{Proof of Theorems}
\begin{proof}[Proof of Theorem \ref{lemma:1}]
Given the assumption of $L$-Lipschitz, together with the approximation $\mathbb{E}_{\boldsymbol{z}\sim \mathbb{\mu_{\mathcal{V}}}}[\mathcal{L}_{\phi^*_{\mathcal{S}}}(\boldsymbol{z})]\approx\mathbb{E}_{\boldsymbol{z}'\sim \mathbb{P}}[\mathcal{L}_{\phi^*_{\mathcal{S}}}(\boldsymbol{z}')]$, the theorem follows directly from Eq. \ref{KR}.
\end{proof}
\begin{proof}[Proof of Theorem \ref{theo:grad}]
    We apply a widely-used lemma for analyzing GD with $G$-smooth functions \cite{gradient_descend} to obtain the inequality
    \begin{equation*}
        H_{\mathcal{S}}(\phi^*_{\mathcal{S}})\le H_{\mathcal{S}}(\phi^1)\le H_{\mathcal{S}}(\phi^0) - \eta^0(1-G\eta^0/2)\|\nabla_\phi H_{\mathcal{S}}(\phi^0)\|^2.
    \end{equation*}
    %Then, we first note that 
    According to $H_{\mathcal{S}}(\phi^0)$'s definition, we note that $H_{\mathcal{S}}(\phi^0)\le \Lambda=\underset{\boldsymbol{z}\in \mathcal{T}}{\max} \,\mathcal{L}_{\phi^0}(\boldsymbol{z})$. Additionally, $\eta^0(1-G\eta^0/2)>0$ since $0<\eta^0<\frac{2}{G}$. Therefore, if we define a constant irrelevant to $\mathcal{S}$ as follows:
    $$\small
    C = \eta^0(1-G\eta^0/2)\cdot \min_{\mathcal{S}\subset \mathcal{T}} \frac{\|\nabla_{\phi} H_{\mathcal{S}}(\phi^0)\|^2}{\frac{1}{|\mathcal{S}|}\sum_{\boldsymbol{z}\in \mathcal{S}} \|\nabla_{\phi}\mathcal{L}_{\phi^0}(\boldsymbol{z})\|}>0,$$
    it allows us to prove Eq. \ref{theo4.2}. 
\end{proof}
\begin{proof}[Proof of Theorem \ref{theo:mi}]  
We prove the theorem by exploiting the dual formulation of $\mathbb{S}(\mathcal{S}) = OT_{\boldsymbol{\mathrm{M}}}(\mu_{\mathcal{S}}, \mu_{\mathcal{V}})$. By definition, we have (we use distribution $\mu$ to directly represent the probability mass vector associated with $\mu$ for simplicity in this proof):
\begin{equation}  
\label{dual_formulation}  
\mathbb{S}(\mathcal{S}) = \max_{\boldsymbol{u} \oplus \boldsymbol{v} \leq \boldsymbol{\mathrm{M}}} \left(\mu_{\mathcal{S}}^T {\boldsymbol{u}} + \mu_{\mathcal{V}}^T \boldsymbol{v}\right) = \mu_{\mathcal{S}}^T \boldsymbol{u}^*(\mathcal{S}) + \mu_{\mathcal{V}}^T \boldsymbol{v}^*(\mathcal{S}), \
\end{equation}  
where $\boldsymbol{u}^*(\mathcal{S}) \in \mathbb{R}^{|\mathcal{T}|}$ and $\boldsymbol{v}^*(\mathcal{S}) \in \mathbb{R}^{|\mathcal{V}|}$ are optimal dual variables satisfying $u_i^*(\mathcal{S}) + v_j^*(\mathcal{S}) \leq M_{ij}$ for all $i, j$. Since $(\mu_{\mathcal{S}})_i$ is nonzero only for $i \in \mathcal{S}$, then $u_i^*(\mathcal{S})$ for $i \notin \mathcal{S}$ can take arbitrary values and does not affect $\mathbb{S}(\mathcal{S})$. This implies that the dual constraint is automatically satisfied for $i \notin \mathcal{S}$ by setting $u_i^*(\mathcal{S})$ to sufficiently small. Consequently, $v_j^*(\mathcal{S}) = \min_{i \in \mathcal{S}} \left(M_{ij} - u_i^*(\mathcal{S})\right).$

Next, we analyze adding a sample $\boldsymbol{z} \notin \mathcal{S}$. Write $ \boldsymbol{p}=\mu_{\mathcal{S}}$ and $ \boldsymbol{q}=\mu_{\mathcal{S} \cup \{\boldsymbol{z}\}}$. Note that $|p_i-q_i|= O(1/|\mathcal{S}|^2)$ for $i\in \mathcal{S}$, and that $|p_{\boldsymbol{z}}-q_{\boldsymbol{z}}|=O(1/|\mathcal{S}|)$. By the Sensitivity Theorem \cite{bertsekas1997nonlinear}, which states that $u_i^*(\mu)$ is continuously differentiable with respect to $\mu$ if $\mu_i>0$, we have $u_i^*( \boldsymbol{p}) \approx u_i^*( \boldsymbol{q})$ for $i \in \mathcal{S}$. Thus $v^*_j( \boldsymbol{q})=\min\left(M_{\boldsymbol{z}j}-u^*_{\boldsymbol{z}}(\boldsymbol{q}), \min_{i\in \mathcal{S}}(M_{ij}-u_i^*(\boldsymbol{q}))\right)\approx\min\left(M_{\boldsymbol{z}j}-u^*_{\boldsymbol{z}}(\boldsymbol{q}), v^*_j(\boldsymbol{p})\right)$. Using \(\min(a, b) = \frac{a + b}{2} - \frac{|a - b|}{2}\), we approximate the change in $\mathbb{S}(\mathcal{S})$:  
\begin{equation}
\small
\nonumber
    \begin{aligned}
        \mathbb{S}(\mathcal{S}\cup\{\boldsymbol{z}\})-\mathbb{S}(\mathcal{S})&=\boldsymbol{q}^T \boldsymbol{u}^*(\boldsymbol{q}) + \mu_{\mathcal{V}}^T \boldsymbol{v}^*(\boldsymbol{q})-\boldsymbol{p}^T\boldsymbol{u}^*(\boldsymbol{p})-\mu^T_{\mathcal{V}}\boldsymbol{v}^*(\boldsymbol{p})\\
        &\approx \frac{1}{|\mathcal{S}|} u^*_{\boldsymbol{z}}(\boldsymbol{q}) + \frac{1}{|\mathcal{V}|} \sum_j (M_{\boldsymbol{z}j}-u^*_{\boldsymbol{z}}(\boldsymbol{q})-v^*_j(\boldsymbol{p}))^-.
    \end{aligned}
\end{equation}
If $u_{\boldsymbol{z}}^*(\boldsymbol{q})$ is replaced by any $t \in \mathbb{R}$, the same analysis yields an inequality (greater than). Therefore, we have 
$$\small\mathbb{S}(\mathcal{S} \cup \{\boldsymbol{z}\})-\mathbb{S}(\mathcal{S}) \approx \sup_t \left\{\frac{1}{|\mathcal{S}|} t + \frac{1}{|\mathcal{V}|} \sum_j (M_{\boldsymbol{z}j} - t - v_j^*(\boldsymbol{p}))^-\right\}.  
$$  
Now consider $\boldsymbol{z} \in \mathcal{S}$. Write $\boldsymbol{r}=\mu_{\mathcal{S} - \{\boldsymbol{z}\}}$, and note by the Sensitivity Theorem that $u_i^*(\boldsymbol{r}) \approx u_i^*(\boldsymbol{p})$ for $i \in \mathcal{S} - \{\boldsymbol{z}\}$. Similarly we have 
{\small\begin{equation}
\nonumber
    \begin{aligned}
        &\mathbb{S}(\mathcal{S})-\mathbb{S}(\mathcal{S}-\{\boldsymbol{z}\})=\boldsymbol{p}^T\boldsymbol{u}^*(\boldsymbol{p})-\boldsymbol{r}^T\boldsymbol{u}^*(\boldsymbol{r})+\mu_{\mathcal{V}}^T\boldsymbol{v}^*(\boldsymbol{p})-\mu_{\mathcal{V}}^T\boldsymbol{v}^*(\boldsymbol{r})\\
        &\approx \frac{1}{|\mathcal{S}|} u^*_{\boldsymbol{z}}(\boldsymbol{p}) + \frac{1}{|\mathcal{V}|} \sum_j (M_{\boldsymbol{z}j}-u^*_{\boldsymbol{z}}(\boldsymbol{r})-v^*_j(\boldsymbol{r}))^-\\
        &\approx \sup_t\left\{\frac{1}{|\mathcal{S}|} t + \frac{1}{|\mathcal{V}|} \sum_j (M_{\boldsymbol{z}j}-t-v^*_j(\boldsymbol{r}))^-\right\}.
    \end{aligned}
\end{equation}}
Thus, we estimate changes in $\mathbb{S}(\mathcal{S})$, completing the proof.
\end{proof}

\subsection{Algorithms} \label{app:alg}

\begin{algorithm}[h!]
\caption{Procedure of GORACS}  \label{alg:1}
\begin{algorithmic}[1]
\State \textbf{Input:} Training set $\mathcal{T}$, validation set $\mathcal{V}$, distance matrix $\boldsymbol{\mathrm{D}}^*\in \mathbb{R}^{|\mathcal{T}|\times |\mathcal{V}|}$, gradient norms $\boldsymbol{g}\in \mathbb{R}^{|\mathcal{T}|}$, parameter $\lambda$, selection budget $n$, exchange candidates $k$, max exchange iterations $T$.
%$\mathcal{T}, \mathcal{V}, \phi, n, \lambda$
\State $\boldsymbol{\mathrm{M}} = (D^*_{ij}-\lambda g_i)_{ij}$; \Comment{POO Cost Matrix for $OT_{\boldsymbol{\mathrm{M}}}$.(\ref{score_ot})}
\State $\mathcal{S}_{\text{gre}} \gets \varnothing$; \Comment{Stage 1: greedy search.}
\While{$|\mathcal{S}_{\text{gre}}| < n$} 
    \State {add $\quad \text{argmin}_{\boldsymbol{z}\not \in \mathcal{S}_{\text{gre}}} \text{Gain}_{\boldsymbol{\mathrm{M}}}(\boldsymbol{z}|\mathcal{S}_{\text{gre}})\quad$ to $\quad\mathcal{S}_{\text{gre}}$;} \Comment{Eq.(\ref{greedy})}
\EndWhile
%\State $iter \gets 0$   \Comment{Stage 2: Refinement with pruning}
\State $\mathcal{S}_1\gets \mathcal{S}_{\text{gre}}$; \Comment{Stage 2: refinement.}
\ForAll{$t \in \{1, 2, \dots, T\}$} 
    \State{%  
    \texttt{\textcolor{gray}{// Efficient computation following Sec. \ref{eff_ot_comp}.}}%  
} 
    \State $s^t = OT_{\boldsymbol{\mathrm{M}}}(\mu_{\mathcal{S}_t}, \mu_{\mathcal{V}})$;
    \State $ \boldsymbol{u}^*_t\hspace{-0.2em}\gets$ optimal dual variables of $OT_{\boldsymbol{\mathrm{M}}}(\mu_{\mathcal{S}_t}, \mu_{\mathcal{V}})$;
    \State $R=\lceil|\mathcal{V}|/|\mathcal{S}_t|\rceil$;
    \ForAll{$\boldsymbol{z}\in \mathcal{T}$}
        \State $f^{\boldsymbol{\mathrm{M}}}_{\boldsymbol{z}j}(\boldsymbol{u}^*_t)=\min_{\boldsymbol{z}_i\in\mathcal{S}:\boldsymbol{z}_i\not =\boldsymbol{z}} (M_{ij}-u^*_{ti})\quad 1\le j\le |\mathcal{V}|$;
        \State $\hat{y}_{\boldsymbol{z}} \hspace{-0.3em}\gets\hspace{-0.4em} R$-th largest value of $M_{\boldsymbol{z}j}-f^{\boldsymbol{\mathrm{M}}}_{\boldsymbol{z}j}(\boldsymbol{u}^*_t)$ ranked by $j$;
        \State $\text{MI}_{\boldsymbol{\mathrm{M}}}(\boldsymbol{z}|\mathcal{S}_t)=F_{\boldsymbol{\mathrm{M}}}(\hat{y}_{\boldsymbol{z}}|\boldsymbol{z}, \mathcal{S}_t)$; \Comment{Eq.(\ref{AI})}
    \EndFor
    \State $Outer\gets \text{top-}k\text{-min}_{\boldsymbol{z}\not \in \mathcal{S}_t}\text{MI}_{\boldsymbol{\mathrm{M}}}(\boldsymbol{z}|\mathcal{S}_t)$  ;\Comment{Outer pruning.}
    \State $Inner\gets \text{top-}k\text{-max}_{\boldsymbol{z} \in \mathcal{S}_t}\text{MI}_{\boldsymbol{\mathrm{M}}}(\boldsymbol{z}|\mathcal{S}_t)$;  \Comment{Inner pruning.}
    \ForAll{$(\boldsymbol{i}, \boldsymbol{o}) \in Inner\times Outer$}
        \State $\mathcal{S}' = \mathcal{S}_t -\{\boldsymbol{i}\} + \{\boldsymbol{o}\}$;
        \If{$OT_{\boldsymbol{\boldsymbol{\mathrm{M}}}} (\mu_{\mathcal{S'}}, \mu_{\mathcal{V}})<s^t$}\Comment{Verifying decrease.}
            \State $\mathcal{S}_{t+1}\gets \mathcal{S}'$; 
            \State \textbf{break}
        \EndIf
    \EndFor
\EndFor
\State \textbf{Output:} $\mathcal{S}_{T+1}$.  
\end{algorithmic}
\end{algorithm}

\begin{algorithm}[h!]  
\caption{Procedure of Label-enhanced GORACS}  
\begin{algorithmic}[1]  
\State \textbf{Input:} Partitioned training set $\mathcal{T}=\cup_{k=1}^K \mathcal{T}_k$, partitioned validation set $\mathcal{V}=\cup_{k=1}^K\mathcal{V}_k$, selection budget $n$, other parameters $\mathcal{P}$ required for Alg \ref{alg:1}.
\For{$k=1$ to $K$}  
    \State $n_k=\lfloor n\cdot |\mathcal{V}_k|/|\mathcal{V}|\rfloor$; \Comment{Per-class budget} 
     \State $\mathcal{S}_k \gets \text{CoresetSelection}(\mathcal{T}_k, \mathcal{V}_k, n_k, \mathcal{P})$; \Comment{Alg.(\ref{alg:1})}  
\EndFor
\State $\mathcal{S} \gets \bigcup_{k=1}^K \mathcal{S}_k$;  
\State \textbf{Output:} $\mathcal{S}$.  
\end{algorithmic}\label{alg:2}
\end{algorithm}

\end{document}